\documentclass{neus2025}

\title[Training with Hard Constraints]{Training with Hard Constraints:\\Learning Neural Certificates and Controllers for SDEs}
\usepackage{times}
\usepackage{enumitem}
\usepackage[colorinlistoftodos]{todonotes}
\usepackage{booktabs}
\newtheorem{problem}{Problem}
\newtheorem{assumption}{Assumption}
\newtheorem{mydef}{Definition}
\newtheorem{myremark}{Remark}

\usepackage{dsfont}
\usepackage{float}
\usepackage{bm}
\usepackage{xcolor}
\usepackage{caption}
\usepackage{multirow}
\usepackage{wrapfig}

\RestyleAlgo{ruled, linesnumbered}

\newcommand{\init}{\mathrm{0}}
\newcommand{\goal}{\mathrm{g}}
\newcommand{\safe}{\mathrm{s}}
\newcommand{\unsafe}{\mathrm{u}}
\newcommand{\reals}{\mathbb{R}}
\newcommand{\prob}{\mathrm{Pr}}
\newcommand{\ra}{\textsc{ra}}
\renewcommand{\P}{\mathrm{P}_{\ra}^\pi}
\newcommand{\px}{\mathbf{x}}
\newcommand{\pw}{\mathbf{w}}
\newcommand{\region}{q}
\newcommand{\partition}{Q}
\newcommand{\relu}{\mathrm{ReLU}}
\renewcommand{\L}{\mathcal{L}^\partition}
\newcommand{\Lbound}{\L_\mathrm{bound}}

\coltauthor{
 \Name{Chun-Wei Kong} \Email{Chun-Wei.Kong@colorado.edu}\\
 \Name{Sebastian Escobar} \Email{Sebastian.Escobar@colorado.edu}\\
 \Name{Ibon Gracia} \Email{Ibon.Gracia@colorado.edu}\\
 \Name{Jay McMahon} \Email{Jay.McMahon@colorado.edu}\\
 \Name{Morteza Lahijanian} \Email{Morteza.Lahijanian@colorado.edu}\\
 \addr University of Colorado Boulder, Colorado, USA}

\begin{document}

\maketitle

\begin{abstract}
Due to their expressive power, neural networks (NNs) are promising templates for functional optimization problems, particularly for reach-avoid certificate generation for systems governed by stochastic differential equations (SDEs). However, ensuring hard-constraint satisfaction remains a major challenge. In this work, we propose two constraint-driven training frameworks with guarantees
for supermartingale-based neural certificate construction and controller synthesis for SDEs.
The first approach enforces certificate inequalities via domain discretization and a bound-based loss, guaranteeing global validity once the loss reaches zero. We show that this method also enables joint NN controller-certificate synthesis with hard guarantees. For high-dimensional systems where discretization becomes prohibitive, we introduce a partition-free, scenario-based training method that provides arbitrarily tight PAC guarantees for certificate constraint satisfaction. Benchmarks demonstrate scalability of the bound-based method up to 5D, outperforming the state of the art, and scalability of the scenario-based approach to at least 10D with high-confidence guarantees.

% Reach-avoid specifications for SDEs can be certified by supermartingale conditions, but enforcing the resulting generator inequalities globally is hard for neural certificates and controllers.
% We present a hard-constraint training framework with a certified objective: we partition the domain, propagate interval bounds through the certificate and, and optimize a loss that upper-bounds worst-case violation. 
% Adaptive refinement/merging improves scalability and the same loss supports joint neural controller--certificate synthesis by differentiating through the closed-loop generator.
% For higher-dimensional regimes where partitioning becomes expensive, we provide a complementary scenario method that synthesizes the certificate’s last layer with PAC-style guarantees.
% Empirically, we certify reach--avoid up to 5D with hard guarantees and up to 10D with high confidence, and synthesize certifiable controllers on several nonlinear systems.
% The codes are available at Github.\footnote{\url{https://github.com/sees9730/Certified-Reach-Avoid-via-Neural-Synthesis}}
\end{abstract}


\section{Introduction}

Empirical performance alone is insufficient for \emph{reach-avoid} control of continuous-time stochastic systems, particularly in \emph{safety-critical} applications. In such scenarios, formal guarantees that the controlled dynamics satisfy the specification are required. Supermartingale theory frames this objective as the search for a certificate function that \emph{satisfies a set of inequalities} over a continuous domain. Constructing such a function leads to a functional optimization problem over a chosen template. Classical approaches, such as sum-of-squares (SoS) programming, yield convex formulations but suffer from limited scalability. Neural networks (NNs) provide a flexible and expressive alternative for parameterizing certificates; however, enforcing the required inequalities globally remains an open challenge. Existing methods rely on post-training NN verification to ensure constraint satisfaction, which limits scalability.
In this work, we focus on these challenges and aim to answer two key questions: (i) \textit{How can hard constraints be encoded directly into NN certificate training while preserving scalability?} and (ii) \textit{How can a NN controller and NN certificate be trained jointly?}

Existing approaches for learning NNs with hard constraints provide only partial answers. 
On one end, constraints are enforced \emph{softly} during training via penalty terms, augmented Lagrangian methods, or checks on a finite set of samples \citep{marquez2017imposing,kotary2024learning}.
While these techniques can reduce violations on the training set, 
% they provide no mechanism to rule out violations elsewhere in the domain.
they offer no guarantees over the entire domain.
% , where rare but catastrophic failures may occur.
% 
% On the other end, formal learner-verifier methods alternate between learning and constraint checking via a verifier, e.g., satisfiability-modulo-theories (SMT), interval reasoning, or discretization with Lipschitz arguments. 
% Despite encouraging results in several settings~\citep{goyal2024deepsade,chatterjee2023learner,vzikelic2023learning,edwards2025general}, their effectiveness typically rely on problem structure (e.g., linear/propositional constraints, polynomial dynamics, or discrete-time evolutions)
% that enables the verifier to produce informative feedback for the learner.
% 
On the other end, learner–verifier frameworks alternate between training and formal constraint checking using tools such as SMT solvers, interval analysis, or discretization with Lipschitz arguments. Although effective in structured settings~\citep{goyal2024deepsade,chatterjee2023learner,vzikelic2023learning,edwards2025general}, these methods typically rely on restrictive assumptions (e.g., linear or polynomial dynamics, propositional constraints, or discrete-time systems) that enable informative verifier feedback.
Extending such methods to reach-avoid certificates for general continuous-time, continuous-space stochastic systems described by Stochastic Differential Equations (SDEs) remains challenging, as the constraints include a differential inequality that couples certificate derivatives with state-dependent drift and diffusion over a continuous domain.

In this work, we take a different route: we directly encode hard constraints into NN training so that constraint satisfaction is guaranteed upon termination. We introduce two complementary approaches.
Inspired by NN robustness methods, e.g.,~\citep{gowal2018effectiveness,xu2020automatic}, our first approach is a \emph{bound-training} framework that constructs a bound-based loss. The method partitions the domain, derives parametric bounds for each constraint on every cell, and aggregates worst-case violations into a single loss. Minimizing this loss enforces hard constraints: once it becomes non-positive, all constraints are guaranteed to hold. To address scalability limitations from discretization, we present adaptive refinement and merging strategies to manage the partition effectively. Furthermore, we extend this framework to joint controller–certificate synthesis by backpropagating through a differentiable computation graph derived from the closed-loop SDE dynamics. Our empirical evaluations show that bound-training significantly outperforms existing methods in scalability, though it reaches practical limits around 5D systems.
To handle higher-dimensional SDEs, we introduce a partition-free alternative based on scenario optimization. This approach trains the NN using state samples and provides \emph{probably approximately correct} (PAC) guarantees: with high confidence, the constraints are guaranteed to hold everywhere except on a set of small measure under the sampling distribution. We show that, by focusing on the last-layer parameters, the resulting problem reduces to a linear program, allowing a large number of samples and thus arbitrarily increasing confidence while shrinking the un-guaranteed sets. Consequently, this approach scales to high-dimensional systems (at least 10D) with arbitrarily-tight PAC guarantees.
We provide our code on Github~\citep{CRANS:github}.

In short, the main contributions are: (i) a bound-training framework for NN certificates for reach-avoid specifications with hard guarantees, (ii) its extension to joint synthesis of NN controllers and certificates, (iii) a highly scalable scenario-based training method for NN certificates with arbitrarily-tight PAC guarantees, and (iv) extensive case studies demonstrating improved verification scalability and control synthesis on increasingly challenging systems.

% \paragraph{Contributions.}
% In summary, we:
% \begin{itemize}[itemsep=0.2pt, topsep=1pt, parsep=0pt, partopsep=0pt]
%     \item propose a bound-training approach that certifies reach--avoid specifications for SDEs with hard guarantees;
%     \item extend it to jointly synthesize neural controllers and certificates, combining expressive neural control with formal guarantees;
%     \item develop a scenario-based certificate synthesis method, providing statistical guarantees for higher-dimensional systems; and
%     \item present extensive case studies demonstrating improved verification scalability, as well as control synthesis examples with increasingly challenging dynamics.
% \end{itemize}

\paragraph{Related Work}
\label{sec:related_work}

\textit{Neural network training with hard constraints. }
Hard constraints are commonly enforced \emph{by construction} in discrete-output tasks (e.g., coherent/semantic constraint layers)~\citep{giunchiglia2024ccn+,giunchiglia2020coherent,ahmed2022semantic}, or via post-processing for linear constraints in generative models~\citep{stoian2024realistic}.
Mixed-integer encodings is another approach~\citep{yang2025scalable}, but the resulting problem size (mixed-integer linear program) can grow quickly with network width/depth; moreover, common ReLU architectures are not twice continuously differentiable, making them unsuitable for SDE supermartingale certificates.

\textit{Certificates for dynamical systems. }
Learner--verifier frameworks 
alternate sample-based learning with verification.
For \emph{discrete-time} stochastic systems, existing methods~\citep{chatterjee2023learner,vzikelic2023learning} are mostly low-dimensional due to state- (and even noise-) space partitioning.
 % and some setups do not treat the full-domain boundary as unsafe, so trajectories may exit the domain before reaching the goal.
For \emph{continuous-time} stochastic systems,~\cite{neustroev2025neural} verifies \emph{reach-avoid-stay} with a fixed controller by checking the bounds over a discretized state space, but examples remain two-dimensional and control synthesis is left open. 
% Using control guidance-barrier functions,~\cite{xue2024reach} poses iterative convex programs for SDE control, but is restricted to affine systems and is demonstrated only on a simple 2D SDE without drift dynamics.

A key design choice in learner--verifier frameworks is the verifier. 
Some works use Satisfiability Modulo-Theories(SMT)-in-the-loop training, e.g.,~\cite{abate2020formal,abate2021fossil,peruffo2021automated} verify safety/stability constraints for \emph{deterministic} systems. 
However, SMT scalability degrades with state dimension, and $\delta$-complete solvers (e.g., dReal) may generate spurious counterexamples that increase training iterations; this issue persists in recent unified reach/avoid/stay synthesis frameworks for non-polynomial dynamics~\citep{edwards2025general}. 

\textit{Statistical Certificates. }
Motivated by the scalability limitations of partition-based methods, alternative approaches employ sampling techniques~\citep{campi2009scenario} to verify certificates~\citep{anand2023formally,salamati2021data,samari2025data,liu2016almost}. While some methods incur sample complexity exponential in the state dimension $n$~\citep{anand2023formally,salamati2021data}, PAC-based approaches achieve sample complexity independent of $n$~\citep{samari2025data}. Such methods guarantee, with high confidence, that the certificate conditions hold everywhere except on a small region whose volume shrinks with the sample size. No guarantee is provided within this region, though the certificate remains valid if it holds there.
Existing work primarily focuses on verification of NNs against given constraints. In contrast, we extend these ideas to \emph{synthesize} NN parameters that satisfy the constraints with PAC guarantees. Specifically, we formulate a linear scenario optimization problem to train the last-layer weights of the certificate, resulting in a highly scalable framework with arbitrarily tight PAC guarantees.

\section{Problem Formulation}
% Let $X \subset \mathbb{R}^n$ be a compact set and let $V_\theta:X\to\mathbb{R}$ be a neural network with parameters $\theta\in\Theta$. Our goal is to train $V_\theta$ so that a collection of inequalities holds \emph{everywhere} on specified regions of $X$, rather than only at sampled points. Importantly, these constraints may involve not only the network output but also its derivatives \emph{of arbitrary order} through differential operators combined with nonlinear functions.\footnote{In the reach-avoid instantiation presented later, the constraints involve derivatives up to second order; however, the proposed method applies to constraints involving higher-order derivatives as well.} A general template is
% \begin{equation}\label{eq:general_hard_constraints}
% \begin{aligned}
% \text{find } & \theta \in \Theta \\
% \text{s.t. } 
% & c_k\!\left(x, V_\theta(x), \{D^\alpha V_\theta(x)\}_{|\alpha|\le r_k}\right) \le 0,
% \quad \forall x \in \Omega_k,\quad k=1,\dots,K,
% \end{aligned}
% \end{equation}
% where each $\Omega_k\subseteq X$ and $D^\alpha V_\theta$ denotes a partial derivative of multi-index $\alpha$, with each constraint allowed to depend on derivatives up to some (possibly different) order $r_k$. This form captures simple range constraints (e.g., $V_\theta(x)\ge 0$), constraints restricted to particular subsets (e.g., boundary or initial conditions), and differential inequalities induced by dynamical systems. Next, we describe a concrete instance that motivates our study.

We consider continuous-time, continuous-space stochastic systems given as SDEs, i.e.,
\begin{equation}\label{eq:sde}
    d\px(t) = f\big(\px(t), u(t)\big)\,dt + g(\px(t))\,d\pw(t),
\end{equation}
where $\px(t) \in \mathbb{R}^n$ is the state at time $t \in \mathbb{R}_{\ge 0}$, $u(t) \in U$ is the control, $U \subset \mathbb{R}^{n_u}$ is compact, $f:\mathbb{R}^n \times U \to \mathbb{R}^n$ is the drift, $g:\mathbb{R}^n \to \mathbb{R}^{n\times m}$ is the diffusion, and $\pw$ is $\mathbb{R}^{n_w}$-valued Brownian motion. A \emph{controller} (policy) is a function $\pi : \mathbb{R}^n \to U$ that maps states to control inputs. We let $\Pi$ be the set of all controllers that satisfy the following assumption:
% \ml{state all the assumptions on the dynamics first and then move to defining reach-avoid def.}

\begin{assumption}[Regular Dynamics and Control]\label{assum:reg_dyn}
    % \ml{move above Def 1}
    The drift and diffusion terms $f$ and $g$, as well as every $\pi\in\Pi$ are globally $L$-Lipschitz continuous, with $L > 0$.
    %
    \iffalse
    We assume common regularity conditions on the drift, diffusion, and controller i.e., there exists a positive constant $C$ such that for all $x \in \mathbb{R}^n$ and $u\in U$% and $u\in\Pi$
    ,
    \[
        |f(x,u)| + |g(x)|_F \leq C(1+|x|),
    \]
    and for all $x,y \in \mathbb{R}^n$,
    \[
        |f(x,u(x))-f(y,u(y))| + |g(x)-g(y)|_F \leq C|x-y|
    \]
    % $f(x, u)$ is $L$-Lipschitz continuous in $x$, uniformly in $u$, and $g(x)$ is $L$-Lipschitz continuous in $x$
    \fi
\end{assumption}

% Under this setting, System~\eqref{eq:sde} admits a unique strong solution~\citep{esfahani2016stochastic}.
% Denote $(X_s^{t,x;\bm{u}})_{s\geq t}$ as the unique strong solution  of System~\eqref{eq:sde} starting from time $t$ at the state $x \in \mathbb{R}^n$ under the control $\bm{u}$.
%\textcolor{blue}{IG: I would also state the assumptions on $U$ that are present in \cite{esfahani2016stochastic}}
Given a controller $\pi\in \Pi$ and an initial state $x_0 \in \reals^n$, SDE~\eqref{eq:sde} has a unique strong solution $\px_{x_0}^{\pi}$, which is a random process whose unique law we denote by $\prob^{\pi}_{x_0}$ \citep{esfahani2016stochastic}. 
In this work, we are interested in the reach-avoid probability of System~\eqref{eq:sde}.

\begin{mydef}[Reach-Avoid Probability]
    \label{def:ra_spec}
    Consider System~\eqref{eq:sde}, a compact set $X \subset \reals^n$ of interest, a compact safe set $X_\safe \subset \text{int}(X)$, 
    where $\text{int}(X)$ is the interior of $X$,
    a compact initial set $X_\init \subseteq \text{int}(X_\safe)$, and a compact goal set $X_\goal \subseteq \text{int}(X_\safe)$, and let $X_\unsafe = X \setminus \text{int}(X_\safe)$ be the unsafe set.
    Given a controller $\pi \in\Pi$, the reach-avoid probability is defined as
    \begin{align}
    \label{eq: reach-avoid prob}
        \P(X_\init,X_\unsafe,X_\goal) = \inf_{x_0 \in X_\init} \prob^{\pi}_{x_0}(\exists\, t \geq 0, \ \px^{\pi}_{x_0}(t)\in X_\goal \text{ and } \forall t' \in [0,t] ,\ \px^{\pi}_{x_0}(t')\not\in X_\unsafe ). 
    \end{align}
\end{mydef}

% \begin{assumption}[Reach-avoid specifications]\label{assum:ra_spec}
%     \ml{don't need this.}
%     $X$ is compact, $X_{\init} \subset \text{int}(X \setminus X_{\unsafe})$, $X_{\goal} \subset \text{int}(X \setminus X_{\unsafe})$, $X_{\unsafe} \subset X$ 
%     \ml{hmm... this last one doesn't make sense.  $\reals^n \setminus X$ should also be $X_\unsafe$. \textbf{Nope, for verification it does not need to be, see Assumption 4}
%     \ml{Well... we make the assumption that the controller already keeps the system inside $X$ for the verification examples.  That's different from saying that outside of $X$ is safe.}
%     \ck{These set definition is the same as AAAI paper}
%     }
%     .
% \end{assumption}

% Computation of this probability is challenging as it requires solving SDE~\eqref{eq:sde} and unfolding the trajectory over time.  Instead, using a Lyapunov-like function called a certificate, we can guarantee a lower bound for $\P$ per following theorem, adapted from~\citep[Theorem 1]{neustroev2025neural}.

Direct computation of this probability is difficult, as it requires solving SDE~\eqref{eq:sde} and propagating the trajectory over time. Instead, using a Lyapunov-like function called \text{certificate}, we can guarantee a lower bound on $\P$ via the following theorem, adapted from \citep[Theorem~1]{neustroev2025neural}.

% \begin{theorem}[Reach-avoid certificate~\citep{neustroev2025neural}]\label{theorem:ra_constraints}
    
%     Let $V:X\rightarrow \mathbb{R}^n$ be a twice continuously differentiable function with respect to state $x$.
%     Given System~\ref{eq:sde} and a scalar $\beta_\ra > 1$,
%     if the following constraints hold:
%     \begin{subequations}\label{eq:ra_constraints}
%     \begin{align}
%     & V(x) \ge 0, && \forall x \in X, \label{eq:main-a}\\
%     & V(x) \le 1, && \forall x \in X_{\init}, \label{eq:main-b}\\
%     & V(x) \ge \beta_{RA}, && \forall x \in X_{\unsafe}, \label{eq:main-c}\\
%     & \mathcal{G}[V](x) < 0, && \forall x \in X \setminus \big(X_{\goal} \cup X_{\unsafe}\big), \label{eq:main-d}
%     \end{align}
%     \end{subequations}
%     where 
%     \begin{equation}\label{eq:generator}
%     \mathcal{G}[\,\cdot\,](x) \triangleq 
%     \sum_{i=1}^{n} f_{i}\big(x, u(x)\big)\,\frac{\partial}{\partial x_i}[\,\cdot\,](x)
%     + \frac{1}{2}\sum_{i=1}^{n}\sum_{j=1}^{n} \big[g(x)g(x)^{\!\top}\big]_{ij}
%     \frac{\partial^2}{\partial x_i \partial x_j}[\,\cdot\,](x).
%     \end{equation}
%     Then the reach-avoid probability in Definition~\ref{def:ra_spec} is lower bounded, i.e.,
%     \[
%         P_\ra \geq \frac{1}{1-\beta_\ra}.
%     \]
% \end{theorem}

\begin{theorem}[Reach-Avoid Certificate~\citep{neustroev2025neural}] \label{theorem:ra_constraints}
    Consider System~\eqref{eq:sde} under controller $\pi \in \Pi$, the corresponding infinitesimal generator
    \begin{equation}\label{eq:generator}
        \mathcal{G}[\,\cdot\,](x) \triangleq 
        \sum_{i=1}^{n} f_{i}\big(x, \pi(x)\big)\,\frac{\partial}{\partial x_i}[\,\cdot\,](x)
        + \frac{1}{2}\sum_{i=1}^{n}\sum_{j=1}^{n} \big[g(x)g(x)^{\!\top}\big]_{ij}
        \frac{\partial^2}{\partial x_i \partial x_j}[\,\cdot\,](x),
    \end{equation}
    and the sets 
    % $X \subset \reals^n$, $X_\safe \subseteq \text{Inter}(X)$, $X_\init \subseteq X_\safe$, $X_\goal \subseteq X_\safe$, and $X_\unsafe = X \setminus X_\safe$ 
    $X$, $X_\init$, $X_\goal$, and $X_\unsafe$ in Definition~\ref{def:ra_spec}. 
    If there exists a twice continuously differentiable function $V:X\rightarrow \mathbb{R}$ and scalar $\beta > 0$ such that

    \vspace{-3mm}
    \begin{subequations}
        \label{eq:ra_constraints}
        \hspace{-3mm}
        \begin{minipage}{0.37\textwidth}
            \begin{align}
                & V(x) \ge 0, && \forall x \in X, \label{eq:main-a}\\
                & V(x) \le 1, && \forall x \in X_{\init}, \label{eq:main-b} 
            \end{align}
        \end{minipage}
        \hspace{8.8mm}
        \begin{minipage}[c]{0.53\textwidth}
            \begin{align}
                & V(x) \ge \beta, && \forall x \in X_\unsafe, \label{eq:main-c}\\
                    & \mathcal{G}[V](x) < 0, && \forall x \in 
                    X \setminus \text{int}(X_\unsafe \cup X_\goal),
                    % X_\safe \setminus X_\goal, 
                    \label{eq:main-d}
            \end{align}
        \end{minipage}\\
    \end{subequations}
    
    \noindent
    then $V$ is called a \emph{reach-avoid certificate}, 
    % guaranteeing
    %the reach-avoid probability in \eqref{eq: reach-avoid prob} 
    and
    it holds 
    that $\P(X_\init,X_\unsafe,X_\goal) \geq 1 - \frac{1}{\beta}.$
\end{theorem}

Intuitively, $V(x)$ is a non-negative function over $X$ with (desirably) large values on $X_{\unsafe}$ ($\beta$ in Condition~\eqref{eq:main-c}) relative to $X_{\init}$ (Condition~\eqref{eq:main-b}). Condition~\eqref{eq:main-d} ($\mathcal{G}[V](x) < 0$) enforces a strict expected decrease of $V(x)$ along the trajectories of System~\eqref{eq:sde}.  Hence, a larger lower bound $\beta$ of $V(x)$ on $X_\unsafe$ increases the likelihood that trajectories initialized in $X_\init$ reach $X_\goal$ before visiting $X_\unsafe$.  
Constructing such a function, however, is non-trivial.

% Importantly, $\mathcal{G}[V]$ combines derivatives of $V$ with the (generally nonlinear) dynamics $f(x,u)$ and diffusion $g(x)$.
% Based on Theorem~\ref{theorem:ra_constraints}, we formulate the neural verification problem:

The first goal of this work is to leverage the expressive power of neural networks (NNs) as universal function approximators to learn a valid $V$ that certifies a given controller.

\begin{problem}[Neural Certificate]\label{probl:verification}
Given System~\eqref{eq:sde} with domain $X$, initial set $X_\init$, unsafe set $X_\unsafe$, and goal set $X_\goal$ as in Definition~\ref{def:ra_spec}, a controller $\pi \in \Pi$, and threshold probability $p_\ra \in (0,1)$, learn a NN $V_{\theta}(x)$ %parameterized by $\theta$, 
that certifies a reach-avoid probability of at least $p_\ra$, i.e., $\P(X_\init, X_\unsafe, X_\goal) \geq p_\ra$.
% a  train a certificate neural network $V_{\theta}(x)$, parameterized by $\theta$, to guarantee a lower bound of the reach-avoid probability $P_\ra$.
%
% and a reach-avoid specification (Definition~\ref{def:ra_spec}),
% find a certification function $V$ such that the certificate constraints in Definition~\ref{def:ra_cert}
% are satisfied either:
% \begin{itemize}
%     \item \emph{everywhere} over the prescribed domain (hard satisfaction), or
%     \item \emph{with confidence} in a probabilistic sense (statistical satisfaction).
% \end{itemize}
\end{problem}

Often, existing controllers, especially deep policies, fail to ensure high reach-avoid probabilities, resulting in repeatedly re-training the controller and re-attempting certification without a clear guidance for how the controller should improve.
The second objective of this work is to address this issue by jointly training a NN controller and a NN certificate, 
enabling the certificate constraints to directly guide controller updates during training. This coupling reduces trial-and-error and steers learning toward controllers that are both effective and certifiable.
% allowing certificate constraints to guide controller updates during training. This coupling can reduce trial-and-error and steer learning toward controllers that are both effective and certifiable.

\begin{problem}[Joint Neural Controller-Certificate Synthesis]\label{probl:control_synthesis}
Given System~\eqref{eq:sde} with domain $X$, initial set $X_\init$, unsafe set $X_\unsafe$, and goal set $X_\goal$ as in Definition~\ref{def:ra_spec}, and threshold probability $p_\ra \in (0,1)$, jointly learn NNs $\pi_{\theta_{\pi}}(x)$ and $V_{\theta}(x)$, respectively parameterized by $\theta_{\pi}$ and $\theta$, that guarantee a reach-void probability of at least $p_\ra$, i.e., $\P(X_\init, X_\unsafe, X_\goal) \geq p_{\ra}$.

\end{problem}

% \paragraph{Challenges \& Approach Overview}
% The main challenge in Problems~\ref{probl:verification} and \ref{probl:control_synthesis} is \emph{not} training neural networks to maximize the lower bound on $\P$. Rather, the difficulty lies in training neural networks that \emph{satisfy the hard constraints} in $\eqref{eq:ra_constraints}$ to ensure certificate validity. Existing methods typically incorporate these constraints into the optimization loss, which results in a training process with \emph{soft} constraints. Regardless of how small the loss becomes, constraint satisfaction is not guaranteed, hence requiring an expensive post-training verification step to check whether the learned function satisfies the constraints.

\textit{Challenges \& Approach Overview. }
The main challenge in Problems~\ref{probl:verification}~\&~\ref{probl:control_synthesis} is \emph{not} training NNs to maximize the lower bound on $\P$, but ensuring that they \emph{satisfy the constraints} in $\eqref{eq:ra_constraints}$.
Our approach therefore focuses on constraint satisfaction: we (i) enforce Eq.~\eqref{eq:ra_constraints} via bound-training that relies on partioning $X$ (Sec.~\ref{sec:hard_sat}) and (ii) complement it with a scenario-based approach that provides PAC guarantees (Sec.~\ref{sec:scenario}) when partition-based certification becomes costly in higher dimensions.

\begin{myremark}
The proposed training methods extend beyond Probs.~\ref{probl:verification}~\&~\ref{probl:control_synthesis}, providing a \emph{general} framework for training (sufficiently smooth) NNs under arbitrary hard constraints of the form in Eq.~\eqref{eq:ra_constraints}.
\end{myremark}

\section{Training with Hard Constraints}
We present two approaches for training with hard constraints:
(i) bound-training with hard guarantees and (ii) sampling-based scenario optimization with PAC guarantees.  We begin with former.

\vspace{-1mm}

\subsection{Bound-Training for Hard Guarantees}\label{sec:hard_sat}
Consider a twice continuously-differentiable NN $V_{\theta}: X \rightarrow \mathbb{R}$ with learnable parameters $\theta$.
It is well established that, using interval arithmetic~\citep{sunaga1958theory}, bounds of $V_\theta$ can be derived over a given compact set $q \subset X$.  Let $\underline{V}_{\theta},\overline{V}_{\theta},:2^{X} \to \reals$ be functions that provide these bounds, i.e.,
\vspace{-1mm}
\begin{equation}
    \label{eq: bounds of V}
    \underline{V}_{\theta}(q) \leq 
    V_{\theta}(x)
    \leq \overline{V}_{\theta}(q) \qquad \qquad \forall x \in \region.
    \vspace{-1mm}
\end{equation}
% 
% where $\underline{V}_{\theta}(\region),\overline{V}_{\theta}(\region)$ denote the lower and upper bounds, respectively.
Note that these bounds can be obtained in analytical forms (parameterized with $\theta$)
using existing tools (e.g., auto-LiRPA~\citep{xu2020automatic}) for common NNs.
Furthermore, the bounds are differentiable w.r.t. $\theta$, as they are computed recursively by applying the NN activation functions to affine combinations of the weights and biases (see~\citep{goyal2024deepsade} and references therein).
Using $\underline{V}_{\theta},\overline{V}_{\theta}$, we derive a bound-based loss function for certificate training of $V_\theta$,
as defined below.

\vspace{-0.5mm}

\begin{mydef}[Bound-based loss function]  \label{def:bound_train_loss}
Let $\partition = \{\region_1, \ldots, \region_{n_Q}$\}, where $\region_i \subseteq X$, be a partition of $X$, i.e., $\cup_{i=1}^{n_Q} \region_i = X$ and $ \region_i \cap \region_j = \emptyset$ for all $i\neq j \in \{1,\ldots,n_Q\}$,
and consider an NN $V_\theta$ with its bounds $\underline{V}_\theta$ and $\overline{V}_\theta$ in Eq.~\eqref{eq: bounds of V}.
Denote by $\overline{\mathcal{G}[V_{\theta}]}(\region_i)$ an upper bound of $\mathcal{G}[V_{\theta}](x)$ for all $x \in \region_i \in \partition$.
% , and the explicit computation of this bound is introduced later.
% 
The \emph{bound-training loss} for the constraints in Eq.~\eqref{eq:ra_constraints} for a given $\beta >0$ is defined as:
% \begin{subequations}
\begin{align}
    \Lbound(V_\theta, \beta) =
    w_{\ge 0} \L_{\ge 0}(V_\theta)
    + w_\init \L_\init(V_\theta)
    + w_\unsafe \L_\unsafe(V_\theta,\beta)
    + w_{\mathrm{gen}} \L_{\mathrm{gen}}(V_\theta), \label{eq:ra_hard_total_loss}
\end{align}
% \end{subequations}
% 
\noindent
where $w_{\ge 0}, w_\init, w_\unsafe, w_{\mathrm{gen}} > 0$ are positive weights, and\\
\begin{subequations}
    \label{eq:eq:ra_hard_loss}
    % \hspace{-5mm}
    \begin{minipage}{0.44\textwidth}
        \label{eq: loss constraints}
        \begin{align}
            & \L_{\ge 0}(V_\theta) = \sum_{\region_i \in \partition} \relu(-\underline{V}_{\theta}(\region_i)), \label{eq:ra_hard_loss-b}\\
            & \L_\init(V_\theta) = 
            \!\!\!\!\!\!\!\!\!
            \sum_{\substack{\region_i \in \partition, \,
            \region_i \cap X_\init \neq \emptyset}} 
            \!\!\!\!\!\!\!\!\!\!\!
            \relu(\overline{V}_{\theta}(\region_i)-1), \label{eq:ra_hard_loss-c} 
        \end{align}
    \end{minipage}
    \hspace{1em}
    \begin{minipage}[c]{0.52\textwidth}
        \begin{align}
            & \L_\unsafe(V_\theta,\beta) = 
            \hspace{-5mm}
            \sum_{\substack{\region_i \in \partition, \, \region_i \cap X_\unsafe \neq \emptyset}} 
            \hspace{-5mm}
            \relu(\beta-\underline{V}_{\theta}(\region_i)), \label{eq:ra_hard_loss-d}\\
            & \L_{\mathrm{gen}}(V_\theta) = 
            \hspace{-5mm}
            \sum_{\substack{\region_i \in \partition, \, \region_i \cap {\bar X} \neq \emptyset}} 
            \hspace{-5mm}
            \relu(\overline{\mathcal{G}[V_{\theta}]}(\region_i) + \varepsilon_{\mathrm{gen}}),  \label{eq:ra_hard_loss-e}
        \end{align}
    \end{minipage}\\
\end{subequations}

\noindent
with
$\relu(y) = \max\{0, y\}$ for $y \in \reals$, 
set
$\bar{X} = X \setminus \text{int}(X_\goal \cup X_\unsafe),$
and small scalar
$\varepsilon_{\mathrm{gen}} > 0$.
\end{mydef}

\vspace{-0.5mm}

Intuitively, the bound-training loss $\Lbound$ measures the extent to which $V_\theta$ violates the certificate constraints in Eq.~\eqref{eq:ra_constraints}, given the bounds $\overline{V}_\theta$, $\underline{V}_\theta$, and $\overline{\mathcal{G}[V_{\theta}]}$. Specifically, the loss terms in Eqs.~\eqref{eq:ra_hard_loss-b}–\eqref{eq:ra_hard_loss-e} upper-bound the violation of constraints \eqref{eq:main-a}–\eqref{eq:main-d}, respectively.
This yields a key guarantee as stated in the following theorem.

\vspace{-0.5mm}

\begin{theorem}[Constraint Satisfaction Guarantees]
\label{prop:hard_sat}
    % \ml{shouldn't this be a theorem?  It's a main result of the paper, no?}
    Given System~\eqref{eq:sde} under controller $\pi \in \Pi$, reach-avoid probability threshold $p_\ra \in (0,1)$, 
    and partition $\partition$ of $X$,
    let $V_{\theta}$ be trained according to 
    $\min_{\theta}\,\Lbound(V_\theta, \frac{1}{1-p_\ra})$,
    where $\Lbound$ is the bound-training loss in Definition~\ref{def:bound_train_loss}.
    If $\Lbound(V_\theta,$ $\frac{1}{1-p_\ra}) = 0$, then $V_{\theta}$ 
    satisfies the constraints in Definition~\ref{theorem:ra_constraints}, i.e., is a valid reach-avoid certificate, and the reach-avoid probability is lower-bounded by $\P(X_\init, X_\unsafe, X_\goal) \geq p_\ra$.
\end{theorem}

\vspace{-1mm}

\begin{proof}
    The proof follows by the fact that $\relu(y) = 0$ iff $y \leq 0$.
\end{proof}

This theorem shows that a bound-training loss $\Lbound$ of zero ensures validity of the certificate. 
The following theorem shows the other way around, i.e., given a strictly satisfying certificate, its associated $\Lbound$ value becomes zero for sufficiently fine partition.

\begin{theorem}\label{prop:consistency}
Assume that, for a given NN $V^\star_{\theta}$ and reach-avoid probability threshold $p_\ra \in (0,1)$,
there exist parameters \(\theta\) such that \(V^\star_\theta\) satisfies the inequalities in
Theorem~\ref{theorem:ra_constraints} \emph{strictly} (i.e., without equalities).
Let \(\{\partition^{(k)}\}_{k \in \mathbb{N}, k\ge 1}\) be a sequence of partitions of \(X\) with mesh diameter
$\mathrm{diam}(\partition^{(k)})\triangleq \max_{\region\in\partition^{(k)}}\sup_{x,y\in\region}\|x-y\|_2$ and
\(\partition^{(k+1)}\) is obtained from \(\partition^{(k)}\) by splitting one or more cells such that
\(\mathrm{diam}(\partition^{(k)})\to 0\) as \(k\to\infty\).
Assume the certified bounds used to define
\(\mathcal{L}^{\partition^{(k)}}_{\mathrm{bound}}\) are sound and consistent under refinement,
then $\lim_{k\rightarrow \infty}\mathcal{L}^{\partition^{(k)}}_{\mathrm{bound}}(V^\star_\theta,\beta)=0$.
\end{theorem}

The proof of Thm.~\ref{prop:consistency} is provided in Appendix~\ref{appendix:proof_hard_bound_consist}.
% 
% By Theorem~\ref{prop:hard_sat}, learning $V_{\theta}(x)$ until the bound-training loss is zero guarantees that the certificate constraints are satisfied everywhere.
% Moreover, we show that feasibility is recovered in the limit as the partition is refined (Proposition~\ref{prop:consistency}).
% 
Theorems~\ref{prop:hard_sat} and~\ref{prop:consistency} establish that the proposed loss function enables reasoning about the soundness of a certificate provided that the partition is sufficiently fine. Based on these results, we design an efficient computational framework. 
% To that end, we first make the following observation: 
%\begin{myremark}
    % No samples $\{x\}_{i=1}^N \in X$ are needed to evaluate the bound loss (Definition~\ref{def:bound_train_loss}).
    We also remark that the proposed bound-training approach does not require samples from $X$ or System~\eqref{eq:sde}.
%\end{myremark}

% 
Before introducing the bound-raining algorithm, we first present a method of bounding $\overline{\mathcal{G}[V_{\theta}]}$ and a refinement strategy that improves training convergence.

\paragraph{Bounding $\mathcal{G}[V_{\theta}]$}
% 
% Given an SDE and a certificate network $V_{\theta}(x)$, we derive the analytical expression of $\mathcal{G}[V_{\theta}]$ (Eq.~\eqref{eq:generator})
% We derive an analytical expression for
% bound $\overline{\mathcal{G}[V_{\theta}]}$
% by applying chain rule recursively through the NN $V_\theta$ (see Appendix~\ref{appendix:gv_derive} for an illustrative example of a fully-connected sigmoid network).
% We use this analytical expression to construct a computation graph $\Phi :X \rightarrow \mathbb{R}$ such that (i) 
% \ml{sounds like magic at the moment.  Need to provide a construction}
% $\Phi(x) = \mathcal{G}[V_{\theta}](x)$, where $\Phi$ shares the same learnable parameters $\theta$, and (ii) the upper bound $\overline{\mathcal{G}[V_{\theta}]}(\region_i)$ can be obtained efficiently via interval bound propagation through $\Phi$.
% Accordingly, the bound-training in 
% % Theorem~\ref{prop:hard_sat} can be written explicitly as
% % \begin{equation}\label{eq:update_verificaiton}
%     % \min_{\theta} \Lbound\big(V_{\theta}(x), \Phi_{\theta}(x), \partition_\ra; \; P_\ra \big).
% % \end{equation}
% Eq.~\eqref{eq:ra_hard_loss-e} becomes computable.
% 
We derive an analytical expression for the bound $\overline{\mathcal{G}[V_{\theta}]}$ by recursively applying the chain rule through the NN $V_\theta$ (see Appendix~\ref{appendix:gv_derive} for an illustrative example using a fully connected sigmoid network). 
This yields a computation graph $\Phi_\theta : X \rightarrow \mathbb{R}$ that ensures $\Phi_{\theta}(x) = \mathcal{G}[V_{\theta}](x)
$ for all $x \in X$ while sharing the same learnable parameters $\theta$.
Accordingly, the upper bound $\overline{\mathcal{G}[V_{\theta}]}(\region_i)$ for each $\region_i \in \partition$ can be obtained via interval bound propagation through $\Phi_{\theta}$, making bound-training loss in Eq.~\eqref{eq:ra_hard_total_loss} computable. 
% for verification (Problem~\ref{probl:verification}), and we minimize the following expression:
Reflecting $\Phi_\theta$ in the notation of $\Lbound$, the training objective becomes:
\begin{equation}\label{eq:update_veri}
    \min_{\theta} \Lbound\big(V_{\theta}, \Phi_{\theta}, \beta).
\end{equation}
\noindent

\paragraph{Adaptive refinement and merging}
Here, we introduce a refinement procedure to enable efficient bound-training based on Theorem~\ref{prop:consistency}.
Given initial partition $\partition$, we categorize its regions according to the constraints, i.e., $\partition_{\geq 0} = \partition$, $\partition_\init = \{\region_i \in \partition \mid \region_i \cap X_\init \neq \emptyset\}$, $\partition_\unsafe = \{\region_i \in \partition \mid \region_i \cap X_\unsafe \neq \emptyset\}$, and $\partition_{\mathrm{gen}} = \{\region_i \in \partition \mid \region_i \cap \big( X \setminus \text{int}\big(X_\goal \cup X_\unsafe) \big) \neq \emptyset \}$.
For each sub-loss function $\L_{j}$ with $j\in\{\geq 0,\\ \, \init, \, \unsafe, \, \textrm{gen} \}$ in Eq.~\eqref{eq: loss constraints},
% each $\partition_{(\cdot)} \in \partition_\ra$, 
we increasingly tighten the bounds on $V_\theta$ and $\mathcal{G}[V_{\theta}]$ by adaptively refining the corresponding partition $\partition_j$.
Specifically, we split each cell $\region \in \partition_{j}$, whose corresponding $\relu$ evaluation is positive, into smaller regions every $k_{\mathrm{refine}}$ epochs.
% , for all $\region_i$ whose evaluation of the $\relu(\cdot,\region_i)$ in the corresponding bound-training loss (Definition~\ref{def:bound_train_loss}) are positive.

One central issue of this refinement is the exponential growth in the number of cells, especially in higher dimensions.
Hence, we adopt two strategies: (i) \emph{top-K refinement} and (ii) \emph{merging} to mitigate this issue.
Firstly, rather than refining all $\region \in Q_j$ with positive $\relu$ values, a top-K rule is applied: for each sub-loss $\L_j$, we rank the cells $\region \in \partition_{j}$ in descending order of their $\relu$ values and refine only the top-K cells.
% with the largest positive values
Secondly, every $k_{\mathrm{merge}}$ epochs, we merge neighboring cells $\region,\region' \in \partition_{j}$ if their corresponding bounds ($\underline{V}_\theta$, $\overline{V}_\theta$, or $\overline{G[V_\theta]}$)
satisfy the constraints with some margins, e.g., adjacent cells $\region,\region' \in \partition_{\geq 0}$ are merged if $\underline{V_{\theta}}(\region),\underline{V_{\theta}}(\region') \geq \alpha_{\text{margin}}$, where $\alpha_{\text{margin}} \gg 0$.

Our empirical evaluations demonstrate the effectiveness of this refinement method, which enables scalability to 5D SDEs. 
We next present a comprehensive and 
unified algorithm that integrates all the aforementioned methods for neural certificate and joint certificate-controller training.

\paragraph{Hard-SAT Algorithm for Neural Certificate and Controller Synthesis}
Theorems~\ref{prop:hard_sat} and \ref{prop:consistency} establish that our bound-training method ensures hard-constraint satisfaction, enabling a simple procedure for joint training of neural controller-certificate.
%Here we show how to extend this method for joint neural controller-certificate synthesis.
Note that the controller only affects the infinitesimal generator $\mathcal{G}$ (Eq.~\eqref{eq:generator}) in Theorem~\ref{theorem:ra_constraints}.
Let $\pi_{\theta_{\pi}}(x)$ denote the initialization of the controller, with parameters $\theta_{\pi}$. Then, we write the bound-based generator constraint in Eq.~\eqref{eq:ra_hard_loss-e} as:
$
    \L_{\mathrm{gen}}(\Phi_{\theta,\theta_{\pi}}) = 
    \sum_{\region_i\in \partition_{\mathrm{gen}}} 
    \relu(\overline{\Phi}_{\theta, \theta_\pi}(\region_i) + \varepsilon_{\mathrm{gen}}),
$
where $\Phi_{\theta, \theta_\pi}$ is the analytical generator of $V_{\theta}$, parameterized by both $\theta$ and $\theta_\pi$, inducing the following bound-based loss to jointly update them:
% \vspace{-1mm}
\begin{equation}\label{eq:update_ctrlsyn}
    \min_{\theta,\theta_{\pi}} \Lbound\big(V_{\theta},  \Phi_{\theta,\theta_{\pi}}, \beta).
\end{equation}

Unlike existing joint synthesis approaches (often formulated for deterministic and/or discrete-time systems) which employ sample-based losses, the bound-training one in Eq.~\eqref{eq:update_ctrlsyn} provides a single, certificate-driven, objective for learning both $V_\theta$ and $\pi_{\theta_{\pi}}$. 
From a multi-network training perspective, optimizing a unified loss often improves the stability of joint training compared to balancing multiple objectives. 
From a certification perspective, our loss explicitly couples controller performance with constraint satisfaction. 
Empirically, this coupling guides the controller not only toward task performance but also toward certification throughout training.
In Algorithm~\ref{alg:hard_sat}, we present a unified pseudocode for neural certificate or joint synthesis with hard constraint satisfaction.

%%%%%%%%%%%%%%%%%%%%%%%%%

\begin{algorithm}[htb!]
\caption{Hard-SAT for Verification or Control Synthesis}\label{alg:hard_sat}
\KwIn{System~\eqref{eq:sde}, sets $X,X_\init,X_\unsafe,X_\goal$, threshold $p_\ra$, $k_{\max}$
}
\KwOut{$\mathrm{UNSAT}$ or $(\mathrm{SAT}, V_{\theta},\pi)$}

Initialize $V_{\theta}$, $\pi$ ($\pi\leftarrow \pi_{\theta_{\pi}}$ for control synthesis), $\beta \gets \frac{1}{1-p_\ra}$, and $\partition_{\geq 0}, \partition_\init, \partition_\unsafe, \partition_\textrm{gen}$

Construct computation graph $\Phi_\theta$ (or $\Phi_{\theta,\theta_\pi}$)

\While{$k \leq k_{\max}$}{
 Compute bounds $\underline{V}_\theta$, $\overline{V}_\theta$, and  $\overline{\Phi}_\theta$ (or $\overline{\Phi}_{\theta,\theta_\pi}$) 
 
 Update $\theta$ by Eq.~\eqref{eq:update_veri} (Verification) or  Update $(\theta,\theta_{\pi})$ by Eq.~\eqref{eq:update_ctrlsyn} (Control Synthesis)
  
  % \If{$\Lbound = 0$}{
  %   \Return{$(\mathrm{SAT}, V_{\theta}, \pi)$}\;
  % }
  \lIf{$\Lbound = 0$}{\Return{$(\mathrm{SAT}, V_{\theta}, \pi)$}}

  \If{$k = k_{\mathrm{refine}}$ or $k = k_{\mathrm{merge}}$}{
    % \If{the number of total discretization exceeds $\region_{\max}$}{
    %     \Return{$\mathrm{UNSAT}$}
    % }
    \lIf{out of memory (OM)}{\Return{$\mathrm{UNSAT}$}}
    Do refinement or merging on $\partition_{\geq 0}, \partition_\init, \partition_\unsafe, \partition_\textrm{gen}$\;
  }
}
\Return{$\mathrm{UNSAT}$}
\end{algorithm}

\noindent
% \input{sections/alg_sample}

%%%%%%%%%%%%%%%%%%%%%%%%%

% \paragraph{A few remarks on Algorithm~\ref{alg:hard_sat}}
During the initialization of $V_{\theta}$ (and $\pi_{\theta_{\pi}}$ for control synthesis) in Algorithm~\ref{alg:hard_sat}, one can warm-start them by standard techniques (e.g., sample-based training or reinforcement learning). In practice, such warm-starting allows appropriate initialization of the networks and could improve the convergence rate of the bound-training. 
Furthermore, Algorithm~\ref{alg:hard_sat} terminates when either $\Lbound = 0$ or the training epoch reaches the compute limit, returning $\mathrm{UNSAT}$.
\begin{myremark}\label{remark:opt_algo}
    To improve performance, we can adopt an incremental learning strategy.
    Initialize Algorithm~\ref{alg:hard_sat} with a small threshold $p \leq p_\ra$. 
    If $\Lbound = 0$, update 
    $p \gets p+\Delta p$ and rerun Algorithm~\ref{alg:hard_sat}.
    Repeat this procedure until $p \geq p_\ra$.
\end{myremark}

\paragraph{Limitations of Hard-SAT Algorithm}
Despite the adaptive refinement and merging strategies, Algorithm~\ref{alg:hard_sat} still relies on state-space partitioning to obtain sound bounds over the entire domain. 
As a result, the number of partitioned cells can grow exponentially with the state dimension, limiting scalability.  Below, we introduce an alternative approach to address this limitation.

\subsection{Scenario-based Training for PAC Guarantees}\label{sec:scenario}

%\ml{I suggest the follow structure for this section:\\- brief overview of scenario\\- formalize scenario optimization\\- theorem on the guarantees.\\The current version is dense and repetitive.}

%\ml{this paragraph needs to be shortened.}
Here, we propose a statistical approach that avoids state-partitioning and relies on randomly sampling a finite set of states $\hat{X} = \{\hat{\px}^{(i)}\in X\}_{i=1}^N$. 
Our method is based on the convex scenario approach, which relaxes the constraints in Eq.~\eqref{eq:ra_constraints} by considering only their evaluation at the sample points in $\hat{X}$. These sampled constraints define the feasible set of a finite-dimensional program, the \emph{scenario program}. Specifically, we construct a convex program on the weights and bias of the last layer of $V_\theta$, denoted $\theta_L$, and on $\beta$. If $N$ is large enough, then our approach guarantees that, with high confidence, the obtained certificate $V_\theta$ (and the associated controller) satisfies the constraints in Eq.~\eqref{eq:ra_constraints} over most of the state space, except for a small region whose volume shrinks with $N$.
%the probability that $\theta^*_N$ violates the constraints of the original robust convex program is at most $\epsilon$. 
%This is a statistical guarantee on the feasibility of $\theta^*_N$ and, when $\theta$ parameterizes the certificate network $V_{\theta}$, it yields a statistical guarantee on the probability that the resulting certificate (and the associated controller) violates the original constraints.
%Intuitively, with high confidence, the synthesized certificate (and the controller) satisfy the constraints in Eq.~\eqref{eq:ra_constraints} over $X\setminus D_\epsilon$, where the measure (volume weighted by $P$) of set $D_\epsilon$ is at most $\epsilon$. 

The scenario approach enables better scalability than the bound-training method in Sec.~\ref{sec:hard_sat} as the sample complexity required to obtain a bound on the volume of the violating region and a confidence level is independent from the dimension of the state space, as formalized below. 
% This is formalized in Theorem~\ref{theorem:scenario_cvx} below.
% \begin{mydef}\label{def:scenario_constraint}
%     given the reach-avoid specification in Definition~\ref{def:ra_spec} with certificate constraints in Eq.~\eqref{eq:ra_constraints},
%     we formulate a single scenario-based constraint to express the constraints in Eq.~\eqref{eq:ra_constraints} as:
%     \begin{multline*}
%     h(\theta, \beta, x) :=\\ -V_\theta(x) + \mathds{1}_{X_\text{init}}(x)(V_\theta(x) - 1) +\mathds{1}_{X_\text{unsafe}}(x)(\beta - V_\theta(x)) + \mathds{1}_{X \setminus (X_\goal \cup X_\unsafe)}(x)(G [V_\theta](x) + \varepsilon_{\mathrm{gen}}).
%     \end{multline*}
% \end{mydef}
%
% Denote by $\theta_L$ the weights and bias of the last layer of the certificate network $V_{\theta}(x)$. 

% We propose to solve Problems~\ref{probl:verification} and~\ref{probl:control_synthesis} with statistical guarantees via a two-stage approach: first, we warm  start both a controller and a certificate 
% % simultaneously 
% using any user-defined method, e.g., soft training by encoding sampled constraints in the loss function. Then, we optimize for $\theta_L$ and $\beta$ via a scenario program, and consider all remaining parameters as fixed.

%\ml{The way the theorem is stated, there is no reason to limit ourselves to just the last layer's parameters. It's not clear where $\theta_L$ show up in the constraints. Why do you claim convex program? What are your assumptions?  the last activation function is convex/linear? can you relate $h$ to the sub-loss function $\L$ in Eq.~\eqref{eq: loss constraints}?}

%
\begin{theorem}\label{theorem:scenario_cvx}
    %\ml{this theorem is too dense}
    Consider System~\eqref{eq:sde} and the pre-trained %certificate and controller 
    networks $V_{\theta}$ and $\pi_{\theta_{\pi}}$. Let $\theta_L \in \mathbb{R}^{d_v}$ denote the weights and bias of the last layer of $V_{\theta}$, whose activation function is linear, and define
    \begin{multline*}
        h(\theta_L, \beta, x) := \relu(-V_\theta(x)) + \mathds{1}_{X_\init}(x)\relu(V_\theta(x) - 1)  \\ + \mathds{1}_{X_\unsafe}(x)\relu(\beta - V_\theta(x)) + \mathds{1}_{X \setminus \text{int}(X_\goal \cup X_\unsafe)}(x)\relu(\mathcal{G} [V_\theta](x) + \varepsilon_{\mathrm{gen}}),
   \end{multline*}
   %\ml{check if it is sufficient for $h < 0$ for all the constraints to hold. Say $V(\bar x) = -1$ for some $\bar{x} \in X_\init$.  Then, $h(\bar{x}) < 0$ but the constraints $V(x) \geq 0$ is violated, no?}
   where $\mathds{1}_{X}(x)$ is the indicator function that returns 1 if $x\in X$ otherwise 0.
   % , and
   % $\bar{X} = X \setminus (X_\goal \cup X_\unsafe)$.
    Given $\epsilon > 0$, $\delta \in (0,1)$, and a probability distribution $P$ on $X$, let $\hat{X} = \{\hat {\px}^{(i)} \in X \}_{i=1}^N$ be a set of random samples from $P$ with $N \ge 2(\log(1/\delta) + d_v)/\epsilon$, and 
    % $\theta^*_{L,N}, \beta_{N}^* \in \arg\min\{ -\beta : h(\theta_L,\beta, \hat {\px}^{(i)}) \le 0 \;\forall i \in \{1, \dots, N\}\}.$
    \begin{equation}\label{eq:rcp1}
    \begin{aligned}
        \theta^*_{L,N}, \beta_{N}^* \in \arg\min_{\theta_L, \beta} -\beta \qquad \textrm{subject to} \qquad h(\theta_L,\beta, \hat {\px}^{(i)}) \le 0 \quad \forall \hat{x}^{(i)} \in \hat{X}.
    \end{aligned}
    \end{equation}
    Then, the optimization problem in Eq.~\eqref{eq:rcp1} is a linear program (LP), and  $P^N\big[P[h(\theta_{L,N}^*, \beta_{N}^*,\px) > 0] \le \epsilon\big] \ge 1- \delta$, where $\px \sim P$. %with confidence no lower than $1-\delta$, $P[h(\theta^*_N, \beta_{\ra,N}^*,x) > 0] \le \epsilon$.
    %\ml{Then, it holds that $P^N\big(P(\forall x \in X, \ h(\theta_{L,N}^*, \beta_{N}^*,x) > 0 \big) \le \epsilon \big) \ge 1- \delta$}
\end{theorem}
The proof of Theorem~\ref{theorem:scenario_cvx} is provided in Appendix~\ref{appendix:proof_scenario_cvx}. Intuitively, Theorem~\ref{theorem:scenario_cvx} guarantees that the synthesized certificate $V_{\theta^*_N}$ (along with the given controller $\pi(x)$) satisfies the constraints in Eq.~\eqref{eq:ra_constraints} %with the reach-avoid probability $\frac{1}{1-\beta_{\mathcal{R}\mathcal{A},N}^*}$
for all $x \in X\setminus D_\epsilon$, where the measure of the violating set is $P(D_\epsilon) \le \epsilon$. 

Based on this result, we propose to solve Problems~\ref{probl:verification} and~\ref{probl:control_synthesis} with statistical guarantees via a two-stage approach: first, warm  start both a controller and a certificate 
% simultaneously 
using any user-defined method, e.g., soft training by encoding sampled constraints in the loss function; then, optimize for $\theta_L$ and $\beta$ via a scenario LP.
% and consider all remaining parameters as fixed.
% 
Note that we only optimize for the weights of the last layer of $V_{\theta}$, but not the remaining ones or controller $\pi_{\theta_{\pi}}$; if those are also made decision variables, $h$ becomes nonlinear, making the results in \citep{campi2009scenario} inapplicable.

%\begin{remark}
    We remark that the scalability of our statistical approach comes at the price of providing weaker guarantees than the ones provided by bound-training (Hard-SAT) in Section~\ref{sec:hard_sat}, as Theorem~\ref{theorem:scenario_cvx} does not rule out existence of a violating set $D_\epsilon$. However, we guarantee that $D_\epsilon$ is very small\footnote{If $P$ is uniform, then the volume ratio of $D_\epsilon$ w.r.t. $X$ is at most $\epsilon$.} and, if the constraints are also satisfied in $D_\epsilon$, then reach-avoid probability $\P \geq 1 - 1/\beta_{N}^*$. 
    % Both approaches are therefore not comparable, and thus 
    Hence, the statistical approach must be seen as a principled way to verify the certificate and controller for high dimensional systems, albeit allowing a small chance of violation.
%\end{remark}

% \iffalse
% \begin{myremark}
%     If $P$ is uniform, then the volume ratio of $D_\epsilon$ w.r.t. $X$ is at most $\epsilon$.
% \end{myremark}
% \fi

% \iffalse
% \begin{myremark}
%     While solving the convex program in Theorem~\ref{theorem:scenario_cvx}, one can initialize $V$ and $u$ networks (if doing control synthesis) by standard training techniques, similar to bound-training.
% \end{myremark}
% \fi
\section{Experiments}

We demonstrate the effectiveness of our approach on several verification and controller synthesis case studies. Specifically, we benchmark neural certificate (Problem~\ref{probl:verification}) on 6 examples and compare to the
% our approaches against 
closest related work \citep{neustroev2025neural}.  
% The results show that
% our Hard-SAT algorithm scales to 5 dimensions and the scenario-based approach scales even higher (10D), whereas the existing method is limited to 2D.
For joint controller-certificate synthesis (Problem~\ref{probl:control_synthesis}), we consider 4 systems of increasing complexity.
Results of the bound-training approach 
% in Sec.~\ref{sec:hard_sat} 
were obtained on an AMD Ryzen 9 7950X CPU with 128 GB RAM, and those of our statistical approach were run on
a single thread of an Intel Core i7 3.6GHz CPU with
32GB of RAM.
Below, we provide a brief discussion of the results, and defer detailed descriptions of the experimental setup, dynamics, plots, and code to the Appendix and Github~\citep{CRANS:github}.

% We evaluate our proposed approaches on several systems on both verification and control synthesis problems. We also compare our results against \citep{neustroev2025neural}

% All experiments ran on an AMD Ryzen 9 7950X (16-core) CPU with 128 GB RAM.
% Training used PyTorch 2.3.1 and Adam optimizer with decaying learning-rate schedule; we used auto-LiRPA~\citep{xu2020automatic} for bound propagation.
% Section~\ref{sec:exp_veri} benchmarks neural verification (Problem~\ref{probl:verification}) on six examples, demonstrating scalability to four dimensions, while the scenario-based approach scales even higher (up to 10-D).
% Section~\ref{sec:exp_syn} demonstrates the joint neural control-certificate synthesis (Problem~\ref{probl:control_synthesis}) on four systems with increasingly complex dynamics.
% We evaluate six benchmarks across 2--4 dimensions: four verification tasks with given controllers (2D/3D/4D Geometric Brownian Motion (GBM) systems and a 2D inverted pendulum system), and three control synthesis tasks where both the certificate function and controller are learned jointly (2D inverted pendulum system, 3D Lorentz system and 3D NASA's XV15 aircraft).
% \se{update paragraph above once all the examples are in}
% \ck{finalize the above paragraph}

\paragraph{Verification}\label{sec:exp_veri}
% First, we study a stochastic inverted pendulum 
% with bounded torque control; Appendix~\ref{appendix:inv_pend} provides the complete specification. 
% We verify the pre-trained controller from~\cite{neustroev2025neural} and compare our results against theirs.
% Next, we consider $n$-dimensional Geometric Brownian Motion (GBM), which has linear drift and \emph{state-dependent} diffusion;
% % \[
% % dx = \left(Ax + \pi(x)\right)dt + g(x) \, dw,
% % \]
% % where $A \in \mathbb{R}^{n \times n}$ is tri-diagonal and $g(x) = 0.2 \text{diag}(x)$ is state-dependent diffusion, and $\pi(x) = -x$ is a given stabilizing controller. 
% The detailed dynamics and reach--avoid specifications are given in Appendix~\ref{appendix:n_GBM}.
% We benchmark verification scalability on GBM because its state dimension can be easily increased without changing the system structure.
%

We first consider the inverted pendulum with bounded torque control (Appendix~\ref{appendix:inv_pend}) under a pre-trained controller from~\citep{neustroev2025neural}.
We then consider $n$-dimensional Geometric Brownian Motion (GBM) with linear drift and state-dependent diffusion (Appendix~\ref{appendix:n_GBM}) to benchmark scalability, as its state dimension can be increased without altering the SDE structure.

\begin{table}[t]
\centering
\vspace{-2mm}
\caption{Verification benchmark, reporting mean runtime (seconds) and mean partition size $|\partition|$ over 5 runs. ``OM'' indicates out-of-memory termination (max $|\partition|$ reported), and ``--'' denotes skipped cases after a lower-dimensional instance hit time/memory limits. 
% Results are shown up to the largest SAT dimension. 
% Entries are marked with ``*'' if only \textcolor{red}{four} runs were attempted due to available time.
}
\label{tab:verification_methods_results}
\resizebox{\textwidth}{!}{%
\begin{tabular}{l|rrr|rrr|clc|clc}
\toprule
\multirow{2}{*}{\textbf{Systems}}
& \multicolumn{3}{c|}{\textbf{Neustroev et al.}}
& \multicolumn{3}{c|}{\textbf{Hard-SAT}}
& \multicolumn{6}{c}{($N=10^5$) \; \textbf{Scenario-based Training} \; ($N=10^6$)}
% & \multicolumn{3}{c}{\textbf{PAC Training} ($10^6$)}
\\
\cmidrule(lr){2-4}
\cmidrule(lr){5-7}
\cmidrule(lr){8-13}
& Time & $|\partition|$ \ \ \ & $p_\ra$
& Time & $|\partition|$ \ \ \ & $p_\ra$
% & Pre/Sat & $N$
& Time & $p_\ra$ & $\epsilon$ 
& Time & $p_\ra$ & $\epsilon$
\\
\midrule
2D Inv Pen & 128.8 & 49,306 & 0.95 & 422.9 & 9,227 & 0.95 & 2.7 & 0.997 & 5.6e-4 & 29.7 & 0.997 & 5.6e-5 \\

2D GBM & 53.5 & 92,911 & 0.95 & 13.9 & 607 & 0.95 & 9.5 & 0.99991 & 1.3e-3 & 92.2 & 0.999991 & 1.3e-4 \\

3D GBM & OM & 434,186 & OM & 640.3 & 7,319 & 0.95 & 9.4 & 0.99991 & 1.3e-3 & 92.7 & 0.99991 & 1.3e-4 \\

4D GBM & -- & -- & -- & 4,241.4 & 61,228 & 0.95 & 2.8 & 0.99991 & 5.6e-4 & 30.6 & 0.99991 & 5.6e-5 \\

5D GBM & -- & -- & -- & 114,282.0 & 342,208 & 0.95 & 2.8 & 0.999 & 5.6e-4 & 30.7 & 0.999 & 5.6e-5 \\

10D GBM & -- & -- & -- & -- & -- & -- & 9.5 & 0.99994 & 1.3e-3 & 94.2 & 0.999994 & 1.3e-4 \\
\bottomrule
\end{tabular}%
}
\vspace{-2mm}
\end{table}

Table~\ref{tab:verification_methods_results} summarizes the verification benchmark of our methods against the state-of-the-art certificate method for SDEs \citep{neustroev2025neural}\footnote{For a fair comparison, we modified the code from~\citep{neustroev2025neural} to remove \emph{stay}-related training and checks.}, which performs soft training of $V_\theta$ and then verifies the constraints in a post-process via partitioning.
For the scenario-based approach, we evaluate two dataset sizes ($N = 10^5$ and $N = 10^6$) with confidences $1-10^{-9}$ (i.e., $\delta = 10^{-9}$)
to examine the trade-off between higher accuracy (smaller $\epsilon$) and computational cost.
% , so all methods are evaluated under the same reach--avoid constraints.
% Our methods (Bound Training and Scenario-based Training) appear in the second and third columns. 
% For bound training, we report computation time (in seconds), number of discrete cells, and certified reach-avoid probability $p_\ra$, averaged over five training trials\footnote{since neural network initialization is random.} 
%\textcolor{blue}{IG: what about stochastic gradient descent?}\ck{that is included in the "initialization"; as we remark that the bound-training does not use any sample, thus no stochastic gradient descent involved.}} for each experiment.

For the 2D systems, our Hard-SAT method requires significantly fewer partitions than\allowbreak~\cite{neustroev2025neural}, with comparable computation times. This efficiency enables bound-training to scale to higher-dimensional GBMs (3D-5D), whereas~\citep{neustroev2025neural} fails to find a valid certificate.

Although outperforming state of the art and providing hard guarantees (i.e., the certificate constraints hold over the entire $X$), bound-training scalability is still limited to 5D. In contrast, the scenario-based approach achieves $p_\ra \approx 1$ across all cases and scales easily to 10D, with the guarantee that, with confidence $1 - 10^{-9}$, all certificate constraints hold over $X \setminus D_\epsilon$. Increasing the sample size from $N = 10^5$ to $10^6$ reduces $\epsilon$ (shrinks $D_\epsilon$) by an order of magnitude, at the cost of a corresponding increase in computation time.
Finally, to demonstrate the effectiveness of optimizing the certificate's last-layer $\theta_L$, we attempted to solve the scenario program using the pre-trained $V_\theta$ without weight optimization. The program was infeasible in every case study, i.e., at least one sample in $\hat{X}$ violated a constraint, indicating that pre-training alone did not result in a valid certificate.

\paragraph{Joint controller-certificate synthesis}\label{sec:exp_syn}

We evaluated 4 systems in an increasing level of complexity.
First, a 2D stochastic inverted pendulum with torque limits, where the controller must swing up from near-down to near-upright while avoiding high angular rates and full rotations. While many learning-based controllers perform well in deterministic settings, few provide guarantees under stochastic noise.
Second, a 2D GBM example with a goal set that excludes the origin, preventing the stabilizing drift dynamics from reaching the goal with high probability.
Third, a stochastic 3D Lorenz system with a reach-avoid task adapted from~\citep{edwards2025general}. The deterministic system is already chaotic~\citep{lorenz2017deterministic}; stochastic noise further complicates control and certification.
Finally, we consider NASA’s XV-15 tilt-rotor aircraft~\citep{ferguson1988mathematical}, focusing on its 3D longitudinal dynamics with bounded inputs, requiring a safe transition from near-hover to forward flight under complex aerodynamics and stochastic disturbances.
Experimental details are in Appendix~\ref{appendix:exp}.

\begin{wraptable}{r}{0.4\textwidth}
  \centering
  \vspace{-5mm}
  \caption{Joint controller-certificate synthesis using Hard-SAT Algorithm. 
  % Metrics: runtime (seconds) and partition size $|\partition|$. 
  }
  % reach-avoid probability threshold $p_\ra$, and Monte Carlo reach-avoid probability $\tilde{\mathrm{P}}_{\ra}^\pi$.
  \label{tab:synthesis_methods_results}
  \vspace{-7mm}
  \scalebox{0.9}{
  \begin{tabular}{lrr}\\
  \toprule
  \textbf{System} & \textbf{Time} (s) & $\mathbf{|\partition|}$ \\
  % & \textbf{NN Controller} & $\boldsymbol{p_\ra}$ & $\tilde{\mathrm{P}}_{\ra}^\pi$ \\
  \midrule
  2D Inv. Pend. & 290.8 & 3,944  \\
  % & Table~\ref{tab:2d_inv_pend_u_nn}  & 0.95 & 1.0  \\
  2D GBM Non-Eq.  & 341.7 & 1,681 \\
  % & Table~\ref{tab:2d_gbm_u_nn} & 0.95 & 1.0 \\
  % 3D Lorentz     & 7,371.1 + 6,359.7 & 40,238 + 28,264 & Table~\ref{tab:lorentz_u_nn}  & 0.95 & 1.00 \\
  3D Lorenz     & 14,175.8 & 68,502 \\
  % & Table~\ref{tab:lorentz_u_nn}  & 0.95 & 1.0 \\
  % 3D XV15       & 28,045.3 + 53,482.7 & 115,332 + 52,079 & Table~\ref{tab:xv15_u_nn} & 0.95 & 1.00\\
  3D XV15       & 82,071.3 & 167,411 \\
  % & Table~\ref{tab:xv15_u_nn} & 0.95 & 1.0
  \bottomrule
  \end{tabular}
  }
  \vspace{-5mm}
\end{wraptable}

Table~\ref{tab:synthesis_methods_results} summarizes the results. In every case study, we set the reach-avoid probability threshold to $p_\ra = 0.95$, and our Hard-SAT algorithm was able to successfully train neural controllers and corresponding certificates. Monte Carlo simulations yielded an empirical reach-avoid probability of 1.0 for all closed-loop systems.

We note that a direct comparison with existing methods is not possible, as they fall outside our setting of joint synthesis of \emph{generic neural} controllers for stochastic reach-avoid problems. For instance, the SMT-based approach of~\citep{edwards2025general} assumes deterministic dynamics, and its extension to SDEs remains an open problem.

Figure~\ref{fig:2D GBM non eq} visualizes the 2D GBM synthesis results; due to page limit, all the other plots are provided in Appendix~\ref{appendix:ctrl_syn_visual}. On the left, the learned certificate rises sharply near $X_\unsafe$ and decreases toward $X_{\goal}$, while the corresponding generator $\mathcal{G}[V_\theta]$ is strictly negative. 
On the right, five simulated rollouts are overlaid on the contour of $V_\theta$, comparing the synthesized controller (pink) against open-loop behavior (aquamarine dashed), with the associated control inputs shown alongside. Without control, the trajectories spiral toward the origin, while the controlled rollouts head directly to $X_\goal$.
\begin{figure}[!h]
  \centering
  \begin{minipage}[t]{0.24\textwidth}
    \centering
\includegraphics[width=\linewidth]{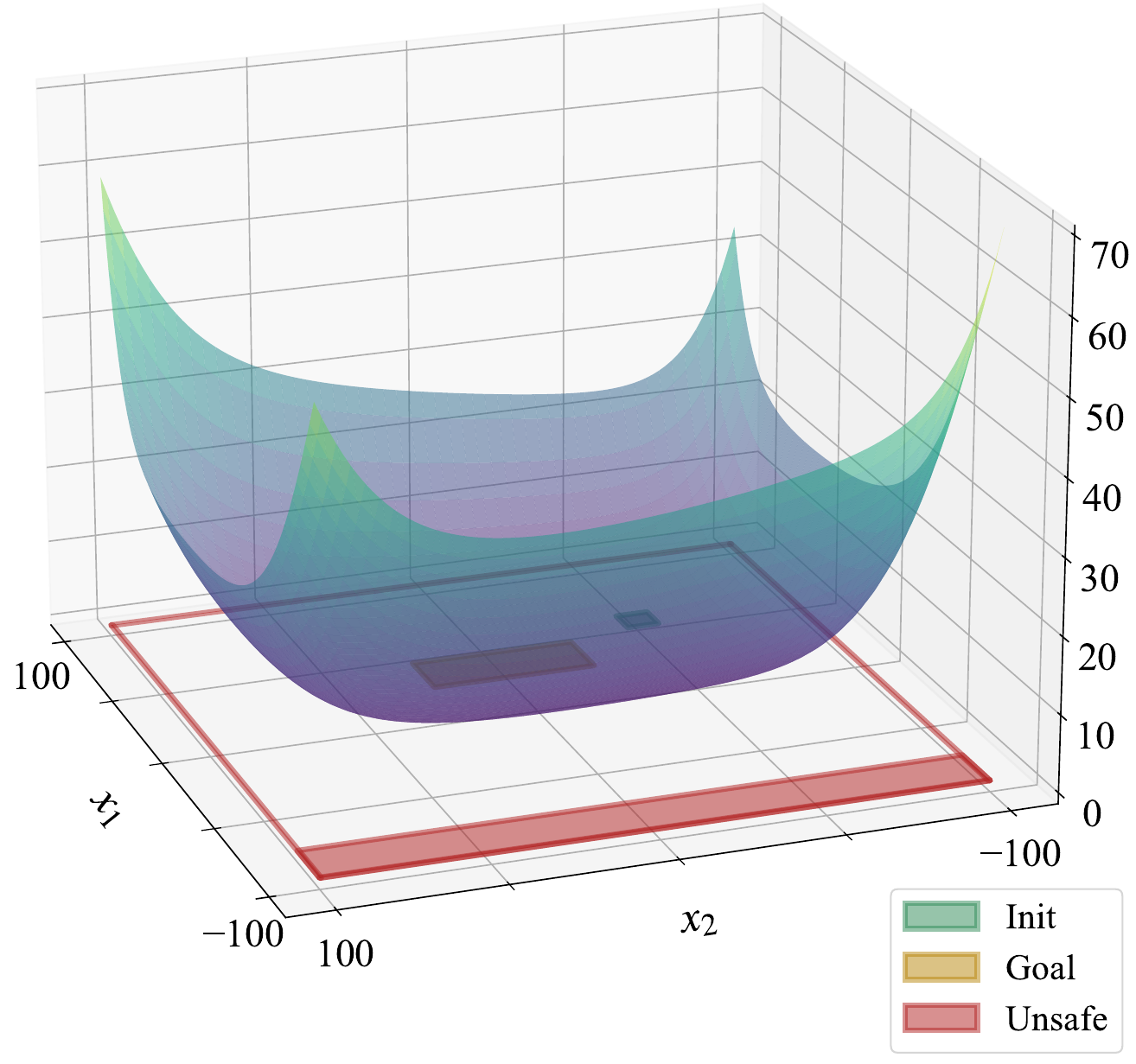}
    \vspace{0.0em}
    {\small (a)}
  \end{minipage}
  \begin{minipage}[t]{0.27\textwidth}
    \centering
\includegraphics[width=\linewidth]{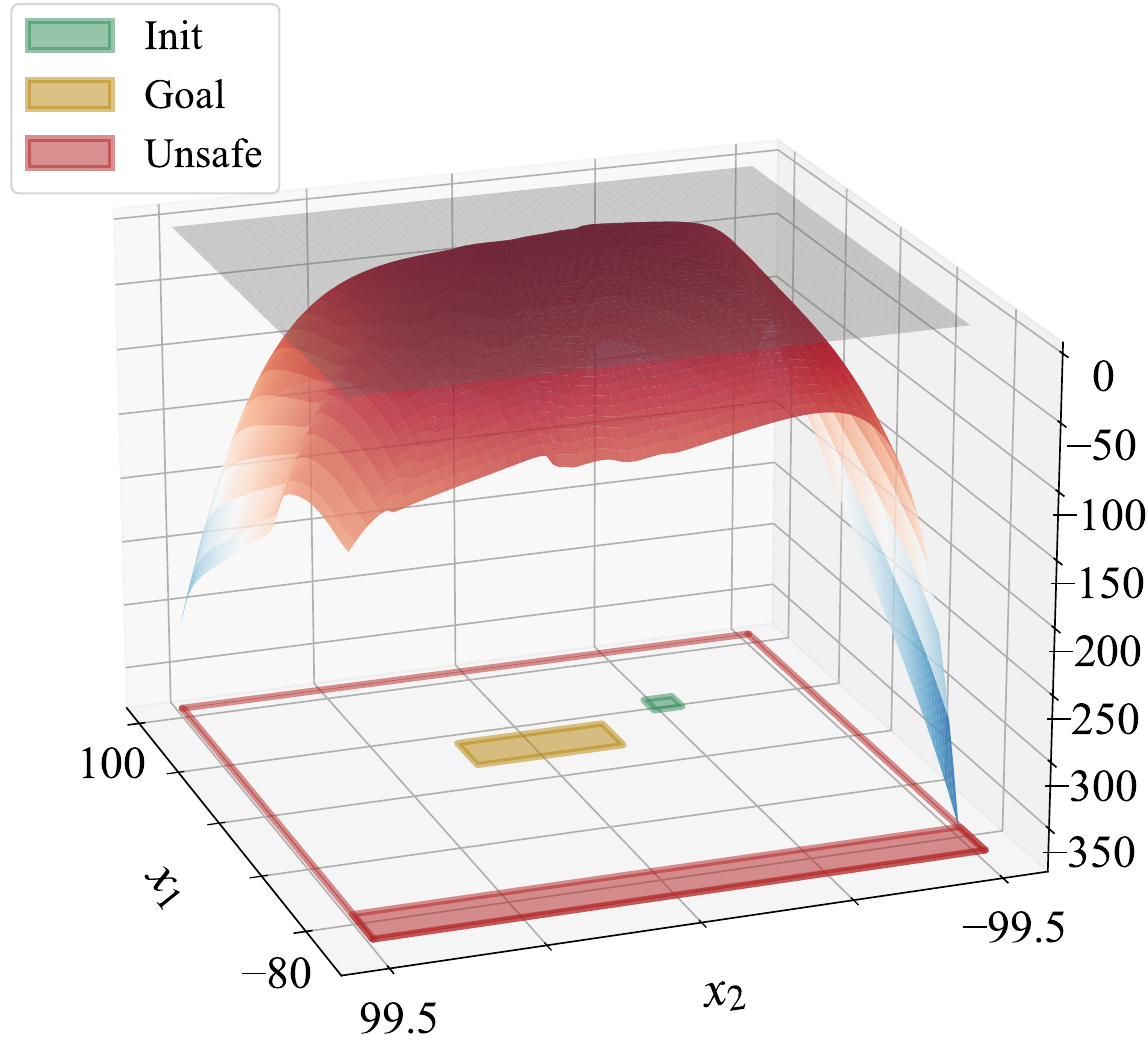}
    \vspace{0.0em}
    {\small (b)}
  \end{minipage}
  \begin{minipage}[t]{0.43\textwidth}
    \centering
\includegraphics[width=\linewidth]{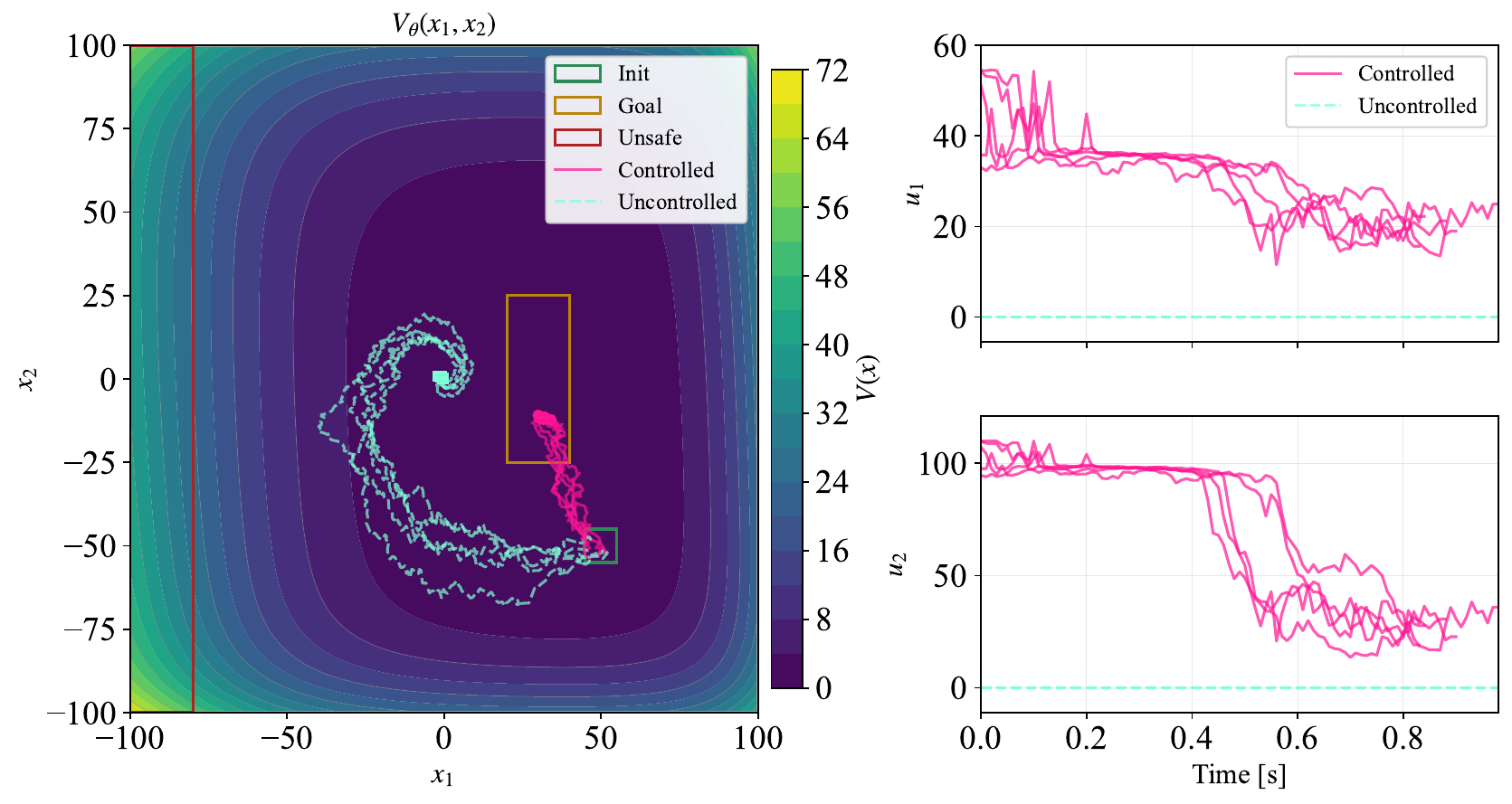}
    \vspace{0.0em}
    {\small (c)}
  \end{minipage}
  \vspace{-2mm}
  \caption{2D GBM synthesis: (a) $V_{\theta}$, (b) $\mathcal{G}[V_{\theta}]$, (c) trajectories (controlled vs.\ uncontrolled).}
  \label{fig:2D GBM non eq}
  \vspace{-2mm}
\end{figure}

\vspace{-1.5em}
\section{Conclusion}
We propose two hard-constrained training frameworks. The first, Hard-SAT, trains a neural reach-avoid certificate for SDEs by enforcing the certificate inequalities through a bound-based loss obtained via domain discretization; once the loss reaches zero, the neural certificate is valid. Given a controller, Hard-SAT verifies SDEs up to 5D, outperforming state of the art. Hard-SAT also enables joint learning of a neural controller and certificate with guaranteed validity upon zero loss.
% , demonstrated up to 3D. 
% For higher-dimensional systems where partitioning is prohibitive, we introduce a scenario-based approach for training, which easily scales to 10D with arbitrarily-tight PAC guarantee. Our results show that, pre-training alone fail to produce a valid certificate, whereas our scenario-based training enabled certification.
For higher-dimensional systems where partitioning becomes prohibitive, we introduce a scenario-based training approach that scales to 10D with tight PAC guarantees. Results show pre-training alone fails to yield a valid certificate, whereas our scenario-based training achieves valid PAC certificate.

% \acks{Acknowledgements text.}

\bibliography{ref}

\appendix
\section{Derivation of analytical $\mathcal{G}[V]$}\label{appendix:gv_derive}
Let $V_{\theta}$ be a two-hidden layer sigmoid network with explicit input and output scaling:
\begin{align*}    x_{\text{norm}} &= x \oslash s_{\text{in}}, \\
z^{(1)} &= W_1 x_{\text{norm}} + b_1, \quad h^{(1)} = \sigma(z^{(1)}), \\
z^{(2)} &= W_2 h^{(1)} + b_2, \quad h^{(2)} = \sigma(z^{(2)}), \\
V_{\theta}(x) &= s_{\text{out}} W_3^T h^{(2)},
\end{align*}
where $\oslash$ denotes the element-wise division, $x \in \mathbb{R}^D$, $W_1 \in \mathbb{R}^{m_1 \times D}$, $W_2 \in \mathbb{R}^{m_2 \times m_1}$, $W_3 \in \mathbb{R}^{m_2}$, $b_1 \in \mathbb{R}^{m_1}$, $b_2 \in \mathbb{R}^{m_2}$, $s_{\text{in}} \in \mathbb{R}^D$, $\mathbf{s}_{\text{out}} \in \mathbb{R}$, $\sigma(z) = 1/(1+e^{-z})$ is the sigmoid activation, and $\theta=\{W_1,b_1,W_2,b_2,W_3\}$.
To compute the bounds on the infinitesimal generator in Eq.~\eqref{eq:generator}, we construct a single feedforward network $\Phi:\mathbb{R}^D \rightarrow \mathbb{R}$ that analytically implements the generator $\Phi(x) = \mathcal{G}[V_{\theta}](x)$. This transforms the problem of bounding derivatives into bounding a standard feedforward network. 
By Eq.~\eqref{eq:generator}, it suffices to derive explicit expression of the first and second derivatives of $V_{\theta}$ with respect to $x$, as detailed below.

To compute $\partial V/\partial x_i$, we apply the chain rule through the network layers. Let $d^{(1)} = \sigma'(z^{(1)})$ and $d^{(2)} = \sigma'(z^{(2)})$ denote element-wise first derivatives of the hidden activations, where $\sigma'(z) = \sigma(z)(1-\sigma(z))$. The gradient with respect to normalized coordinates is:
\[
    \frac{\partial V}{\partial x_{\text{norm},i}} = s_{\text{out}} \sum_{j=1}^{m_2} W_3[j] \cdot d^{(2)}_j \sum_{k=1}^{m_1} W_2[j,k] \cdot d^{(1)}_k \cdot W_1[k,i],
\]
where $W[i,j]$ denotes the $i$-$j$ entry of the matrices $W_1$ and $W_2$, and $W[k]$ denotes the $k$-th element of the vector $W_3$.
Applying the input scaling chain rule gives:
\[
    \frac{\partial V}{\partial x_i} = \frac{1}{s_{\text{in},i}} \frac{\partial V}{\partial x_{\text{norm},i}} = \frac{s_{\text{out}}}{s_{\text{in},i}} \sum_{j=1}^{m_2} W_3[j] \cdot d^{(2)}_j \sum_{k=1}^{m_1} W_2[j,k] \cdot d^{(1)}_k \cdot W_1[k,i]
\]

To compute $\partial^2 V/\partial x_i^2$, let $q^{(1)} = \sigma''(z^{(1)})$ and $q^{(2)} = \sigma''(z^{(2)})$ denote element-wise second derivatives, where $\sigma''(z) = (1-2\sigma(z))\sigma'(z)$. Define the intermediate quantity $S_{j,k,i} \triangleq W_2[j,k] \cdot d^{(1)}_k \cdot W_1[k,i]$, allowing us to write:
\[
    \frac{\partial V}{\partial x_i} = \frac{s_{\text{out}}}{s_{\text{in},i}} \sum_{j=1}^{m_2} W_3[j] \cdot d^{(2)}_j \sum_{k=1}^{m_1} S_{j,k,i}
\]
Differentiating with respect to $x_i$ using the product rule yields two contributions. First, differentiating the outer layer activation derivative $d^{(2)}_j$ gives $\partial d^{(2)}_j/\partial x_i = (q^{(2)}_j/s_{\text{in},i}) \sum_{k=1}^{m_1} S_{j,k,i}$, contributing the cross term:
\[
    \frac{s_{\text{out}}}{s_{\text{in},i}^2} \sum_{j=1}^{m_2} W_3[j] \cdot q^{(2)}_j \left(\sum_{k=1}^{m_1} S_{j,k,i}\right)^2
\]
Second, differentiating the inner layer activation derivative $d^{(1)}_k$ within $S_{j,k,i}$ gives $\partial d^{(1)}_k/\partial x_i = (q^{(1)}_k/s_{\text{in},i}) W_1[k,i]$, contributing the direct term:
\[
    \frac{s_{\text{out}}}{s_{\text{in},i}^2} \sum_{j=1}^{m_2} W_3[j] \cdot d^{(2)}_j \sum_{k=1}^{m_1} W_2[j,k] \cdot q^{(1)}_k \cdot W_1[k,i]^2
\]
Summing both contributions yields the total second derivative:
\[
    \frac{\partial^2 V}{\partial x_i^2} = \frac{s_{\text{out}}}{s_{\text{in},i}^2} \left[ \sum_{j=1}^{m_2} W_3[j] \cdot q^{(2)}_j \left(\sum_{k=1}^{m_1} S_{j,k,i}\right)^2 + \sum_{j=1}^{m_2} W_3[j] \cdot d^{(2)}_j \sum_{k=1}^{m_1} W_2[j,k] \cdot q^{(1)}_k \cdot W_1[k,i]^2 \right].
\]
% The infinitesimal generator combines these derivatives with the drift $f_{\pi}$ and diffusion $g_{\sigma}$:
% \[
%     \Phi(x) = \sum_{i=1}^{D} f_{i}(x) \frac{\partial V}{\partial x_i} + \frac{1}{2}\sum_{i=1}^{D}\sum_{r=1}^{k} g_{\sigma,r}^2(x,i) \frac{\partial^2 V}{\partial x_i^2}
% \]
% All operations form a computational graph with no trainable parameters beyond those of $V$. Auto-LiRPA efficiently propagates IBP bounds through this graph, yielding rigorous bounds $[\underline{\Phi}(C), \overline{\Phi}(C)]$ for any cell $C$. During training, gradients flow through $\Phi$ to update $V$'s parameters.
\section{Proof}

\subsection{Proof of Theorem~\ref{prop:consistency}}\label{appendix:proof_hard_bound_consist}

\begin{proof}
Fix a certificate $V^\star_{\theta}$ satisfying Theorem~\ref{theorem:ra_constraints}
\emph{strictly} with $\beta=\frac{1}{1-p_\ra}$.
Let \(\{\partition^{(k)}\}_{k\ge 1}\) be a sequence of partitions of \(X\), where
\(\partition^{(k+1)}\) is obtained from \(\partition^{(k)}\) by splitting one or more cells such that
\(\mathrm{diam}(\partition^{(k)})\to 0\) as \(k\to\infty\).

By continuity on compact sets ($X,X_\init,X_\unsafe$ and $\bar{X}$) and strict satisfaction, each inequality in
Theorem~\ref{theorem:ra_constraints} holds with a positive margin. Define
\[
\varepsilon \triangleq \min\Big\{
\min_{x\in X} V^\star_\theta(x),\
\min_{x\in X_{\init}} (1-V^\star_\theta(x)),\
\min_{x\in X_{\unsafe}} (V^\star_\theta(x)-\beta),\
\min_{x\in \bar X} (-\mathcal{G}[V^\star_\theta](x))
\Big\} > 0.
\]
Then
$V^\star_{\theta}\ge \varepsilon$ on $X$, $V^\star_{\theta}\le 1-\varepsilon$ on $X_{\init}$,
$V^\star_{\theta}\ge \beta+\varepsilon$ on $X_{\unsafe}$, and
$\mathcal{G}[V^\star_{\theta}]\le -\varepsilon$ on $\bar X$.

By continuity, these strict inequalities extend with margin $\varepsilon/2$ to open neighborhoods of
$X_{\init}$, $X_{\unsafe}$, and $\bar X$, and since $\mathrm{diam}(\partition^{(k)})\to 0$,
for all sufficiently large $k$, every cell intersecting one of these sets is contained in the corresponding neighborhood.
% This literally means: there's an is an open neighborhood of $X_\init$, denoted as $N(X_{\init})$, on which the inequality holds with margin $\varepsilon/2$ (smaller margin), i.e., $X_\init \subset N(X_{\init})$.
% Then by basic lemma, we know: For all sufficiently large $k$, every cell intersecting $X_\init$ is contained in $N(X_{\init}).$
By soundness and consistency of the certified bounds under refinement, together with continuity,
$\underline V^\star_\theta(\region)$, $\overline V^\star_\theta(\region)$, and
$\overline{\mathcal{G}[V^\star_\theta]}(\region)$ become tight as $\mathrm{diam}(\region)\to 0$.
Fix any $0<\varepsilon_{\mathrm{gen}}<\varepsilon/2$. Then, for all sufficiently large $k$
and all $\region\in\partition^{(k)}$,
\[
\underline V^\star_\theta(\region)\ge 0,\quad
\overline V^\star_\theta(\region)\le 1 \ \text{if }\region\cap X_{\init}\neq\emptyset,\quad
\underline V^\star_\theta(\region)\ge \beta \ \text{if }\region\cap X_{\unsafe}\neq\emptyset,
\]
and
\[
\overline{\mathcal{G}[V^\star_\theta]}(\region)+\varepsilon_{\mathrm{gen}}\le 0
\ \text{if }\region\cap \bar X\neq\emptyset.
\]
Therefore every ReLU term in \eqref{eq:ra_hard_loss-b}--\eqref{eq:ra_hard_loss-e} is zero for all sufficiently large $k$,
and hence
\[
\lim_{k\to\infty}\mathcal{L}^{\partition^{(k)}}_{\mathrm{bound}}(V^\star_{\theta},\beta)=0.
\]
\end{proof}

\subsection{Proof of Theorem~\ref{theorem:scenario_cvx}}\label{appendix:proof_scenario_cvx}
\begin{proof}
Express $V_{\theta_L}(x) = \sum_{j = 1}^{m_2} \theta_L^j \phi_j(x)$, where the feature terms $\phi_j$ depend on the architecture of the certificate network and the weights of all but the last layer. 
Also express $\mathcal{G}[V_{\theta_L}](x) = \sum_{j = 1}^{m_2} \theta_L^j \psi_j(x)$, where the terms $\psi_j$ depends on the certificate neural network, its weights besides the last layer, as well as the SDE dynamics. 
Since these expressions are linear in $\theta$, the following program
%
% Source - https://tex.stackexchange.com/a
% Posted by barbara beeton
% Retrieved 2026-01-26, License - CC BY-SA 3.0

\begin{equation}
\label{eq:rcp}
\begin{aligned}
\min_{\theta_L, \beta_{\mathcal{R}\mathcal{A}}} \quad & -\beta_{\mathcal{R}\mathcal{A}}\\
\textrm{s.t.} \quad & h(\theta_L,\beta_{\mathcal{R}\mathcal{A}}, x) \le 0 \:,\forall x\in X,
\end{aligned}
\end{equation}
is convex. Furthermore, since its feasible set is equivalent to the constraints in Eq.~\eqref{eq:ra_constraints}, any feasible solution to Problem~\eqref{eq:rcp} is a valid certificate ensuring a  reach-avoid probability $1-1/\beta_{\mathcal{R}\mathcal{A}}$, which the problem aims to maximize. Let $P$ be a probability distribution on $X$ and $\{\hat{\px}^{(i)}\}_{i=1}^N$ be a set of $N$ samples from $P$. Consider the following convex optimization Problem~\eqref{eq:rcp}
\begin{equation}
\label{eq:scp}
\begin{aligned}
\min_{\theta_L, \beta_{\mathcal{R}\mathcal{A}}} \quad & -\beta_{\mathcal{R}\mathcal{A}}\\
\textrm{s.t.} \quad & h(\theta_L,\beta_{\mathcal{R}\mathcal{A}}, \hat{\px}^{(i)}) \le 0 \:,\forall i \in \{1, \dots, N\},
\end{aligned}
\end{equation}
where the robust constraint in Problem~\eqref{eq:rcp} is relaxed to hold only on the set of $N$ samples. Note that, by the properties of the ReLU function, the feasible set is equivalent to a set of linear inequalities, thus making the problem linear \cite{boyd2004convex}. Let $\theta^*_{L,N}$ and $\beta_{\text{RA},N}^*$ be optimizers to Problem~\eqref{eq:scp}. Given $\epsilon >0$ and $\delta \in (0,1)$, the theory of the convex scenario approach (\cite{campi2009scenario}) states that if the number of samples $N$ is bigger than $2(\log(1/\delta) + d)/\epsilon$, then the probability $P[h(\theta^*_{L,N}, \beta_{\text{RA},N}^*,x) > 0]$ that the optimizers violate the constraint $h$ when $x$ is sampled according to $P$ is not higher than $\epsilon$, with confidence at least $1-\delta$ over the choice of the $N$ samples. %Intuitively, this means that, with confidence no lower than $1-\delta$, the synthesized certificate $V_{\theta^*_{L,N}}$ satisfies the constraints in (REF) for all $x \in X$, except for a region $D_\epsilon$ whose measure $P(D_\epsilon)$ is upper bounded by $\epsilon$.
\end{proof}
\section{Experiment Details}\label{appendix:exp}

\subsection{2D Stochastic Inverted Pendulum}\label{appendix:inv_pend}
The SDE is:
\[
    dx = \begin{bmatrix}
        x_2 \\
        \frac{g}{L} \sin x_1 + \frac{M \pi(x) - b x_2}{mL^2}
    \end{bmatrix}dt + \begin{bmatrix}
        0 \\ \sigma
    \end{bmatrix}dw,
\]
where $x_1$ is the angle, $x_2$ is the angular velocity, the diffusion is $\sigma=2$, $g=9.81$ is the gravity constants, $L=0.5$ is the pendulum length, $m=0.15$ is the ball mass, $b=0.1$ is the friction coefficient, $\pi: \mathbb{R}^2 \rightarrow [-1,1]$ outputs the normalized control torque, and $M=6$ is the maximum torque.
For verification, the reach--avoid specification is:
\begin{align*}
    & X = [-2\pi, 2\pi] \times [-20,20],\; X_\init = [\frac{3\pi}{4}, \frac{5\pi}{4}] \times [-1,1],\\
    & X_\goal = [-\frac{\pi}{2}, \frac{\pi}{2}] \times [-4,4],\; X_\safe = \text{Inter}(X) \setminus \Big(  [-2\pi, -\frac{3\pi}{2}] \times [-20,-10] \Big) \cup \Big( [\frac{3\pi}{2}, 2\pi] \times [10,20] \Big),
\end{align*}
with 95\% probability~\citep{neustroev2025neural}. The controller, given by~\cite{neustroev2025neural}, is a multilayer perceptron (MLP) with three hidden layers (64 neurons) and Tanh activation.

\paragraph{2D Stochastic Inverted Pendulum Control Synthesis}
The architecture of the neural controller is given in Table~\ref{tab:2d_inv_pend_u_nn}.
\begin{table}[htbp!]
    \centering
    \begin{tabular}{lccc}
        \toprule
        \textbf{Layer Connection} & \textbf{Type} & \textbf{\# Neurons} & \textbf{Activation Function} \\
        \midrule
        Input Layer $\rightarrow$ Hidden Layer 1  & Fully Connected & 8  & Tanh \\
        Hidden Layer 1 $\rightarrow$ Output Layer  & Fully Connected & 1  & Tanh \\
        \bottomrule
    \end{tabular}
    \caption{Neural controller of the 2D stochastic inverted pendulum.}
    \label{tab:2d_inv_pend_u_nn}
\end{table}

\subsection{Geometric Brownian Motion (GBM)}\label{appendix:n_GBM}
The SDE of an $n$-dimensional GBM is:
\[
dx = \left(Ax + \pi(x)\right)dt + g(x) \, dw,
\]
where $A \in \mathbb{R}^{n \times n}$ is tri-diagonal: 
\[
A_{ij}=
\begin{cases}
-0.5, & i=j,\\
1, & j=i+1,\\
-1, & i=j+1,\\
0, & \text{otherwise},
\end{cases}
\]
and $g(x) = 0.2 \text{diag}(x)$ is state-dependent diffusion, and $\pi(x) = -x$ is the given stabilizing controller for the verification problem.
\paragraph{2D GBM}
% \[
%     A = \begin{bmatrix}
%         -0.5 & 1 \\
%         -1 & -0.5
%     \end{bmatrix}
% \]
The reach--avoid specification is given by:
\begin{align*}
    & X = [-100,100]^2,\; X_\init = [45,55] \times [-55, -45],\;
    X_\goal = [-25,25]^2,\; \\
    & X_\safe = \text{Inter}(X) \setminus \Big( [-100,-80] \times [-100, 100] \Big)
\end{align*}
with 95\% probability~\citep{neustroev2025neural}.

\paragraph{2D GBM Control Synthesis}
The architecture of the neural controller is given in Table~\ref{tab:2d_gbm_u_nn}.
\begin{table}[htbp!]
    \centering
    \begin{tabular}{lccc}
        \toprule
        \textbf{Layer Connection} & \textbf{Type} & \textbf{\# Neurons} & \textbf{Activation Function} \\
        \midrule
        Input Layer $\rightarrow$ Hidden Layer 1  & Fully Connected & 16  & Tanh \\
        Hidden Layer 1 $\rightarrow$ Output Layer  & Fully Connected & 2  & N/A \\
        \bottomrule
    \end{tabular}
    \caption{Neural controller of the 2D GBM.}
    \label{tab:2d_gbm_u_nn}
\end{table}
The reach--avoid specification has a goal set that does not include the origin:
\begin{align*}
    & X = [-100,100]^2,\; X_\init = [45,55] \times [-55, -45],\;
    X_\goal = [20.0, 40.0] \times [-25,25],\; \\
    & X_\safe = \text{Inter}(X) \setminus \Big( [-100,-80] \times [-100, 100] \Big).
\end{align*}

\paragraph{3D GBM}
% \[
%     A = \begin{bmatrix}
%         -0.5 & 1 & 0 \\
%         -1 & -0.5 & 1 \\
%          0 & -1 & -0.5
%     \end{bmatrix}
% \]
The reach-avoid specification is given by
\begin{align*}
    & X = [-100,100]^3,\;
    X_\init = [45,55] \times [-55, -45] \times [50,60], \\
    & X_\goal = [-25,25]^3,\;
    X_\safe = \text{Inter}(X) \setminus \Big( [-100,-80] \times [-100, 100] \times [-100, -80] \Big),
\end{align*}
with $\mathcal{P}_\ra = 0.95$.

\paragraph{4D GBM}
% \[
%     A = \begin{bmatrix}
%         -0.5 & 1 & 0 & 0 \\
%         -1 & -0.5 & 1 & 0 \\
%          0 & -1 & -0.5 & 1 \\
%          0 & 0 & -1 & -0.5
%     \end{bmatrix}
% \]
The reach-avoid specification is given by
\begin{align*}
    & X = [-100,100]^4,\;
    X_\init = [45,55] \times [-55, -45] \times [50,60] \times [45, 55], \\
    & X_\goal = [-25,25]^4,\;
    X_\safe = \text{Inter}(X) \setminus \Big( [-100,-80] \times [-100, 100] \times [-100, -80]^2 \Big),
\end{align*}
with $\mathcal{P}_\ra = 0.95$.

\paragraph{5D GBM}
% \[
%     A = \begin{bmatrix}
%         -0.5 & 1 & 0 & 0 & 0 \\
%         -1 & -0.5 & 1 & 0 & 0\\
%          0 & -1 & -0.5 & 1 & 0\\
%          0 & 0 & -1 & -0.5 & 1\\
%          0 & 0 & 0  & -1 & -0.5
%     \end{bmatrix}
% \]
The reach-avoid specification is given by
\begin{align*}
    & X = [-100,100]^5,\;
    X_\init = [45,55] \times [-55, -45] \times [50,60] \times [45, 55]^2, \\
    & X_\goal = [-25,25]^5,\;
    X_\safe = \text{Inter}(X) \setminus \Big( [-100,-80] \times [-100, 100] \times [-100, -80]^3 \Big),
\end{align*}
with $\mathcal{P}_\ra = 0.95$.

\paragraph{10D GBM}
% \[
% A \in \mathbb{R}^{10 \times 10}, \qquad
% A_{ij}=
% \begin{cases}
% -0.5, & i=j,\\
% 1, & j=i+1,\\
% -1, & i=j+1,\\
% 0, & \text{otherwise}.
% \end{cases}
% \]
The reach-avoid specification is given by
\begin{align*}
    & X = [-100,100]^10,\;
    X_\init = [45,55] \times [-55, -45] \times [50,60] \times [45, 55]^7, \\
    & X_\goal = [-25,25]^10,\;
    X_\safe = \text{Inter}(X) \setminus \Big( [-100,-80] \times [-100, 100] \times [-100, -80]^8 \Big),
\end{align*}
with $\mathcal{P}_\ra = 0.95$.

\subsection{3D Stochastic Lorenz~\citep{lorenz2017deterministic}}
The SDE is
\begin{align*}
    dx = \begin{bmatrix}
        -10 x_1 + 10 x_2 \\
        x_1(28 - x_3) - x_2 \\
        x_1 x_2 - \frac{8}{3} x_3
    \end{bmatrix}dt + \begin{bmatrix}
        0.1 \\ 0.1 \\ 0.1
    \end{bmatrix}dw.
\end{align*}
The reach--avoid specification (adapted from~\cite{edwards2025general}) is:
\begin{align*}
    & X = [-6,6]^3,\;
    X_\init = [-1,1]^3, \\
    & X_\goal = [-0.3,0.3]^3,\;
    X_\safe = [-5.5,5.5]^3.
\end{align*}
The architecture of the neural controller is a linear feedback, given in Table~\ref{tab:lorentz_u_nn}.
\begin{table}[htbp!]
    \centering
    \begin{tabular}{lccc}
        \toprule
        \textbf{Layer Connection} & \textbf{Type} & \textbf{\# Neurons} & \textbf{Activation Function} \\
        \midrule
        Input Layer$\rightarrow$ Output Layer  & Fully Connected & 3  & N/A \\
        \bottomrule
    \end{tabular}
    \caption{Neural (linear feedback) controller of the 3D stochastic Lorenz system.}
    \label{tab:lorentz_u_nn}
\end{table}

\subsection{XV15 Aircraft~\citep{choi2025data}}\label{appendix:xv15}
We consider the 3D longitudinal dynamics of the XV15 aircraft. The state is $x = [v, \gamma, \beta]$, where $v$ is the airspeed, $\gamma$ is the flight path angle and $\beta$ is the rotor tilt angle.
The dynamics is controlled by three inputs: $u = [T, \alpha, \delta]$, where $T \in [0,T_{\max}]$ is the rotor thrust, $\alpha \in [-\alpha_{\max}, \alpha_{\max}]$ is the angle of attack and $\delta \in [-\delta_{\max}, \delta_{\max}]$ is the rotor tilting rate.
The SDE is:
\[
    dx = \begin{bmatrix}
        -g \sin \gamma + \frac{1}{m}\Big(\cos(\alpha+\beta)T -D(v,\alpha,\beta) \Big) \\
        -g \frac{\cos \gamma}{v} + \frac{1}{m}\Big(\frac{\sin(\alpha+\beta)}{v}T + \frac{L(v,\alpha,\beta)}{v} \Big) \\
        \delta
    \end{bmatrix} dt + 
    \begin{bmatrix}
        0.5 \\ 0.1(\pi/180) \\ 0.1(\pi/180)
    \end{bmatrix}dw,
\]
where $L(v,\alpha,\beta)$ and $D(v,\alpha,\beta)$ are the aerodynamic lift and drag, respectively, and $m$ is the aircraft mass.
The reach--avoid specification is:
\begin{align*}
    & X = [0.5, 100] \times [-20 \text{ deg}, 20 \text{ deg}] \times [0, 90 \text{ deg}],\\
    & X_\init = [28, 32] \times [8.5 \text{ deg}, 10.5 \text{ deg}] \times [58 \text{ deg}, 62 \text{ deg}], \\
    & X_\goal = [65, 85] \times [-2 \text{ deg}, 10 \text{ deg}] \times [25 \text{ deg}, 35 \text{ deg}], \\
    & X_\safe = \Big( [1, 99.5] \times [-19 \text{ deg}, 19 \text{ deg}] \times [1 \text{ deg}, 89 \text{ deg}] \Big).
\end{align*}
The architecture of the neural controller is given in Table~\ref{tab:xv15_u_nn}. The output Sigmoid activation is to ensure $T \in [0, T_{\max}]$, and the other Tanh activations ensure $\alpha$ and $\beta$ respect their output bounds, respectively.
\begin{table}[htbp!]
    \centering
    \begin{tabular}{lccc}
        \toprule
        \textbf{Layer Connection} & \textbf{Type} & \textbf{\# Neurons} & \textbf{Activation Function} \\
        \midrule
        Input Layer $\rightarrow$ Hidden Layer 1  & Fully Connected & 64  & Tanh \\
        Hidden Layer 1 $\rightarrow$ Output Layer  & Fully Connected & 3  & [Sigmoid, Tanh, Tanh] \\
        \bottomrule
    \end{tabular}
    \caption{Neural controller of the 3D XV15 aircraft.}
    \label{tab:xv15_u_nn}
\end{table}

\subsection{Training and Tuning Procedure}
Throughout the experiments, we use the same architecture for neural certificate function given in Table~\ref{tab:v_nn} with the number of neurons reported in Table~\ref{tab:v_nn_sizes}.
% --- Table 1: shared architecture template ---
\begin{table}[htbp!]
    \centering
    \begin{tabular}{lccc}
        \toprule
        \textbf{Layer Connection} & \textbf{Type} & \textbf{\# Neurons} & \textbf{Activation Function} \\
        \midrule
        Input Layer $\rightarrow$ Hidden Layer 1  & Fully Connected & $h_0$  & Sigmoid \\
        Hidden Layer 1 $\rightarrow$ Hidden Layer 2  & Fully Connected & $h_1$  & Sigmoid \\
        Hidden Layer 2 $\rightarrow$ Output Layer  & Fully Connected & 1  & N/A \\
        \bottomrule
    \end{tabular}
    \caption{Neural certificate architecture for all experiments.}
    \label{tab:v_nn}
\end{table}

\begin{table}[htbp!]
    \centering
    \begin{tabular}{lll}
        \toprule
        \textbf{Group} & \textbf{Experiment} & \textbf{$(h_0,h_1)$} \\
        \midrule
        \multicolumn{3}{l}{\textbf{Verification}} \\
        \midrule
        & 2D Inv. Pend. & $(64, 16)$ \\
        & 2D GBM & $(64, 64)$ \\
        & 3D GBM & $(64, 64)$ \\
        & 4D GBM & $(64, 16)$ \\
        & 5D GBM & $(64, 16)$ \\
        & 10D GBM & $(64, 64)$ \\
        \midrule
        \multicolumn{3}{l}{\textbf{Control Synthesis}} \\
        \midrule
        & 2D GBM Non-Eq. & $(64, 64)$ \\
        & 2D Inv. Pend. & $(64, 64)$ \\
        & 3D Lorenz & $(64, 64)$ \\
        & 3D XV15 & $(64, 64)$ \\
        \bottomrule
    \end{tabular}
    \caption{Experiment list grouped by task, with the number of neurons $(h_0,h_1)$ used for the neural certificate in each case.}
    \label{tab:v_nn_sizes}
\end{table}
Below, we describe the adaptive refinement strategy used in practice.

\paragraph{Cell initialization and adaptive refinement.}
Each region is initialized with a prescribed number of cells. If a region’s loss stagnates above zero, this indicates that the current discretization is insufficient to certify the constraint. In such cases, we increment the number of cells for that region and restart training. During training, adaptive refinement is applied periodically, and its frequency is a key tuning parameter. Early in training, refinement occurs less frequently to avoid excessive computation. As training progresses and region losses approach zero, refinement frequency is increased to uncover cells that are already SAT but not yet provably certified. 
% In some cases, subdividing cells allows the network to satisfy the constraint for regions that were previously not certified, while regions that remain difficult are handled by increasing the initial number of cells in subsequent runs. 
The timing and selection of refinement are determined empirically based on monitoring the bound-based loss.

\paragraph{Cell merging.}
Adjacent cells that satisfy the constraints sufficiently, i.e., with respect to some merging margins, are merged. 
The choice of such merging margin is a critical tuning parameter. 
If the margin is too loose, few or no cells will be merged, and the potential benefits of reduced memory usage and faster training are not realized. Conversely, if the margin is too tight, cells that were marginally SAT may be merged, drastically altering the loss landscape and effectively discarding progress achieved during training. In practice, this margin is selected empirically; A good rule of thumb is to start with sufficiently large margins, then reduce them gradually if needed. Figure~\ref{fig:2D GBM non eq discret} illustrates how adaptive refinement and cell merging alter the cell partition, comparing the discretization at the initial and 4,000 epochs.

\begin{figure}[!h]
  \centering
  \begin{minipage}[t]{0.45\textwidth}
    \centering
\includegraphics[width=\linewidth]{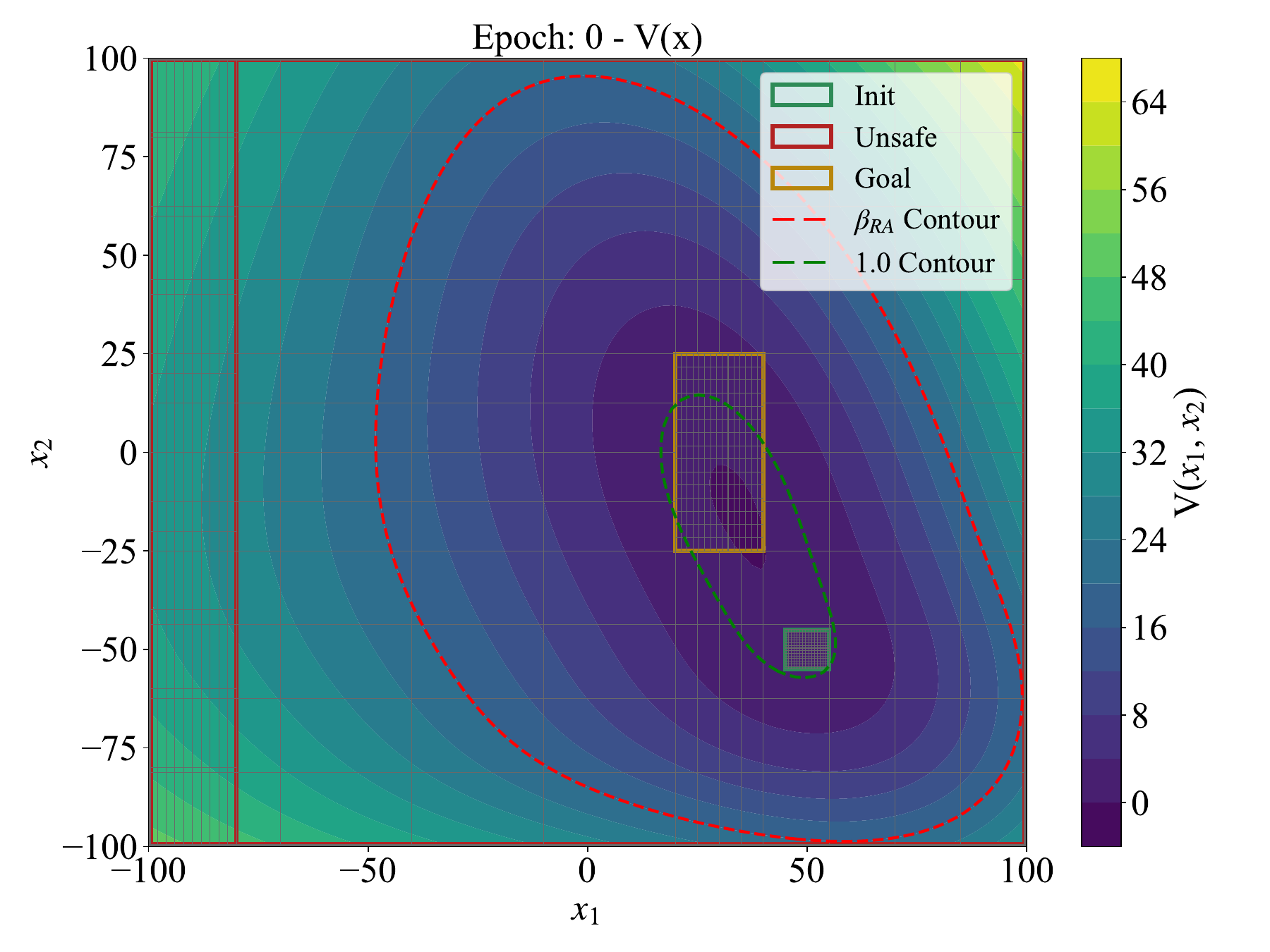}
    \vspace{0.0em}
    {\small (a)}
  \end{minipage}
  \begin{minipage}[t]{0.45\textwidth}
    \centering
\includegraphics[width=\linewidth]{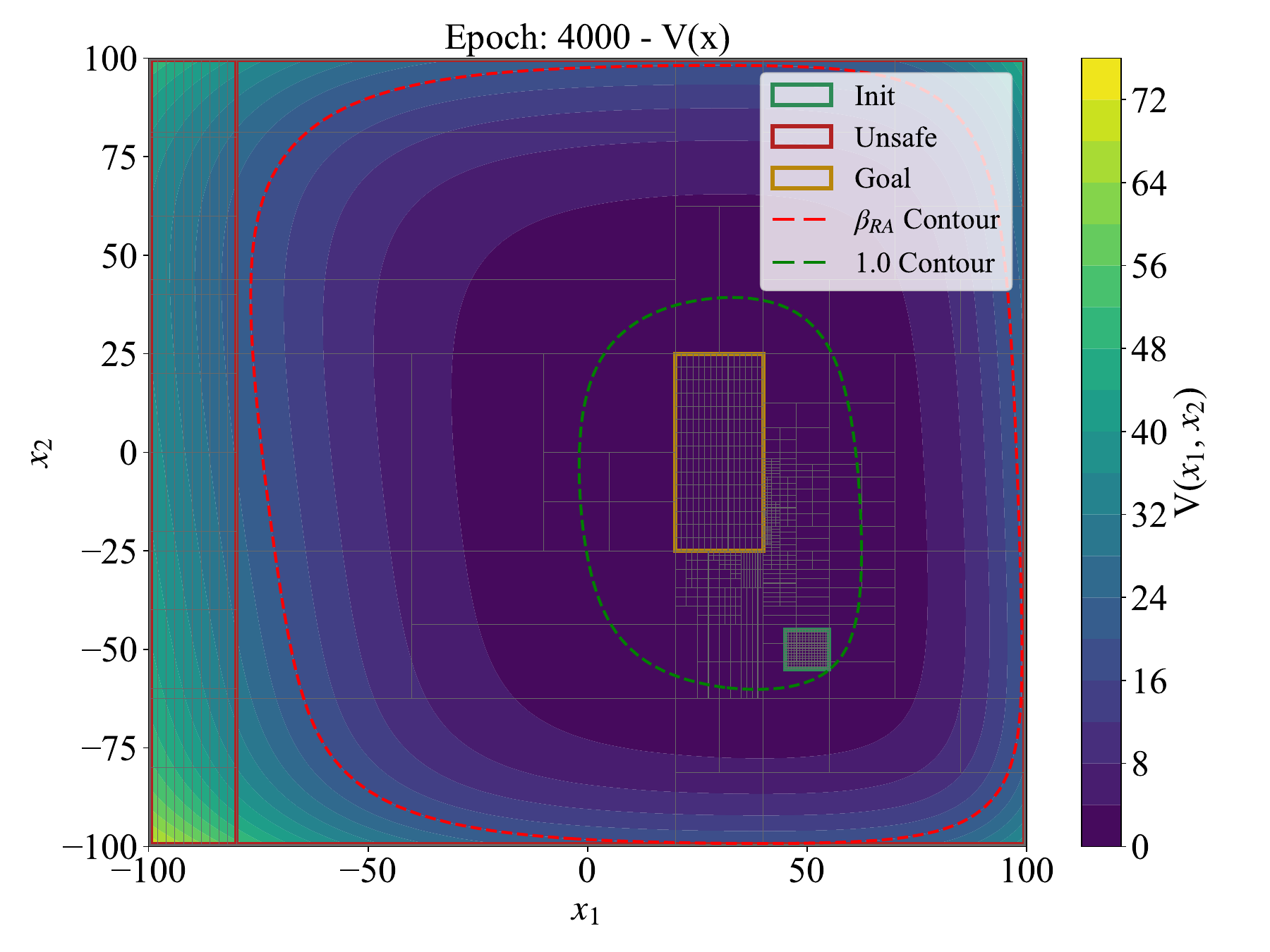}
    \vspace{0.0em}
    {\small (b)}
  \end{minipage}
  \caption{2D GBM synthesis: (a) Cell partition at the end of warm-start training, before bound training; (b) Cell partition after 4000 epochs of bound training, shortly before achieving SAT.}
  \label{fig:2D GBM non eq discret}
\end{figure}

\paragraph{Network architecture and scaling.}
Feedforward networks with two hidden layers of initially equal width are used for the neural certificate $V_{\theta}$. If constraints are satisfied, reducing the width of the second hidden layer can accelerate training.
% without losing feasibility. 
Input and output scaling parameters, $\mathbf{s}_{\mathrm{in}} \in \mathbb{R}^D$ and $\mathbf{s}_{\mathrm{out}} \in \mathbb{R}$ (seen in Appendix~\ref{appendix:gv_derive}), affect training via two mechanisms. First, gradients scale as $\partial V / \partial x \propto \mathbf{s}_{\mathrm{out}} / \mathbf{s}_{\mathrm{in}}$, where larger $\mathbf{s}_{\mathrm{out}}$ accelerates convergence but may induce instability. Second, normalized inputs $x / \mathbf{s}_{\mathrm{in}}$ control bound tightness: larger $\mathbf{s}_{\mathrm{in}}$ tightens IBP bounds but may saturate activations and vanish gradients. 
Since $\mathbf{s}_{\mathrm{in}} \in \mathbb{R}^D$, we define the effective ratio as $\mathbf{s}_{\mathrm{out}} / \max(\mathbf{s}_{\mathrm{in}})$. 
Through experimentation, we select $\mathbf{s}_{\mathrm{out}} / \max(\mathbf{s}_{\mathrm{in}})$ in the range $0.2$–$1.0$, balancing gradient magnitude and bound tightness for stable training. In Eq.~\ref{eq:ra_hard_total_loss}, all loss weights are set to $1$.

\paragraph{Scenario-Based Training Setup}
% \ml{move the following to the appendix}
The state samples are distributed according to a weighted distribution that samples uniformly from $X_\init, X_\goal$ and $X_\unsafe$ each with probability $0.1$ and from the rest of $X$ with probability $0.7$. 
Conditional to each region, the distribution is uniform.
This choice ensures that our dataset contains samples from all regions. 
We note that our approach holds regardless of the sampling distribution.
% Furthermore, to highlight the need to optimize for the weights of the last layer of the certificate, we attempted to verify a controller and certificate by solving the scenario program in Eq.~\eqref{eq:rcp1} but fixing the weights. This resulted in infeasibility of the programs for every case study in Table~\ref{tab:verification_methods_results}.%, thus highlighting the need of optimizing over the weights for our approach to be successful.

%%% \ck{This loss plot may cause confusion due to the "outside" loss; we did not define this in the text.}
% \begin{figure}[!h]
%     \centering
%     \includegraphics[width=01.0\linewidth]{figures/gbm_veri_training_progress/loss_history.pdf}
%     \caption{Illustration of the bound-training loss vs epochs of the 2D GBM verification.}
%     \label{fig:placeholder}
% \end{figure}

\subsection{Control-Certificate Synthesis Visualization}\label{appendix:ctrl_syn_visual}
Here, we present additional visualization for the control synthesis results in Sec.~\ref{sec:exp_syn}.
Figure~\ref{fig:inv_pend_syn_results_comparison} presents the stochastic inverted pendulum synthesis results, comparing synthesized-controlled and uncontrolled rollouts in each panel. The left shows angle snapshots, the middle shows the certificate contour with both trajectories overlaid, and the right shows the corresponding torque inputs for five rollouts. Without control, the pendulum does not swing up and instead stays near the downward equilibrium.

\begin{figure}[htbp!]
    \centering
    \includegraphics[width=01.0\linewidth,height=0.8\textheight,keepaspectratio]{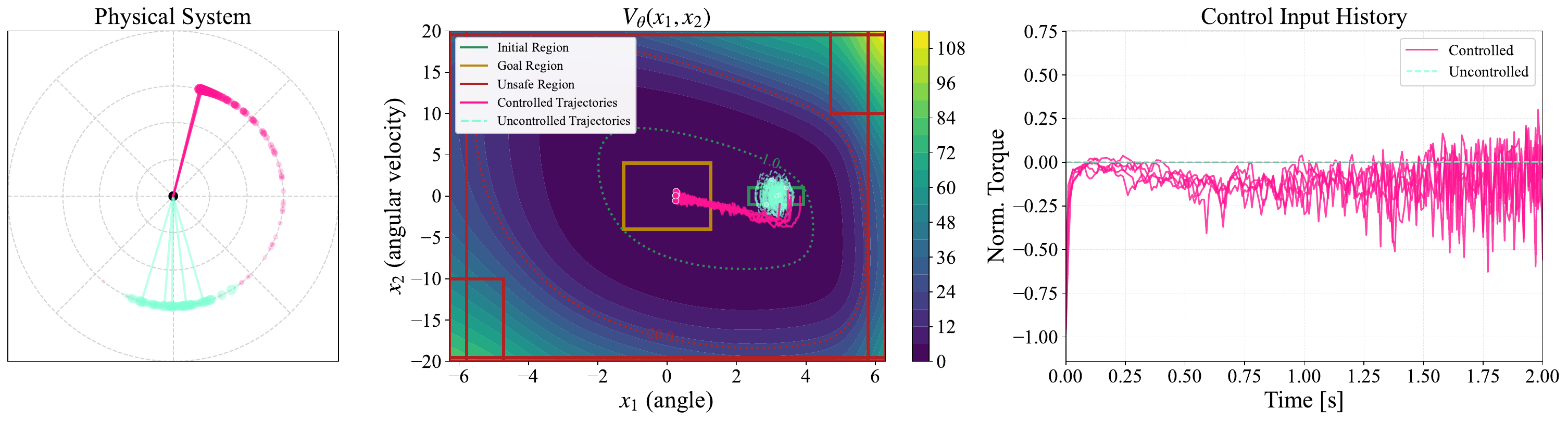}
    \caption{Stochastic inverted pendulum synthesis: angle snapshots (left), phase-plane trajectories overlaid on $V_{\theta}$ (middle), and torque inputs (right), comparing controlled and uncontrolled rollouts.}   \label{fig:inv_pend_syn_results_comparison}
\end{figure}
The synthesis results of Lorenz system are shown in Fig.~\ref{fig:lorenz_syn_results}, comparing open-loop trajectories (aquamarine dashed) with trajectories under the synthesized controller (pink). The left panels plot the 3D phase evolution, and the right panels show the corresponding 2D projections overlaid on the certificate contours. Without control, trajectories leave the safe set (red boxes); with control, they remain safe and are steered into the small goal set containing the origin.
\begin{figure}[htbp!]
    \centering
\includegraphics[width=1.0\linewidth,height=1.0\textheight,keepaspectratio]{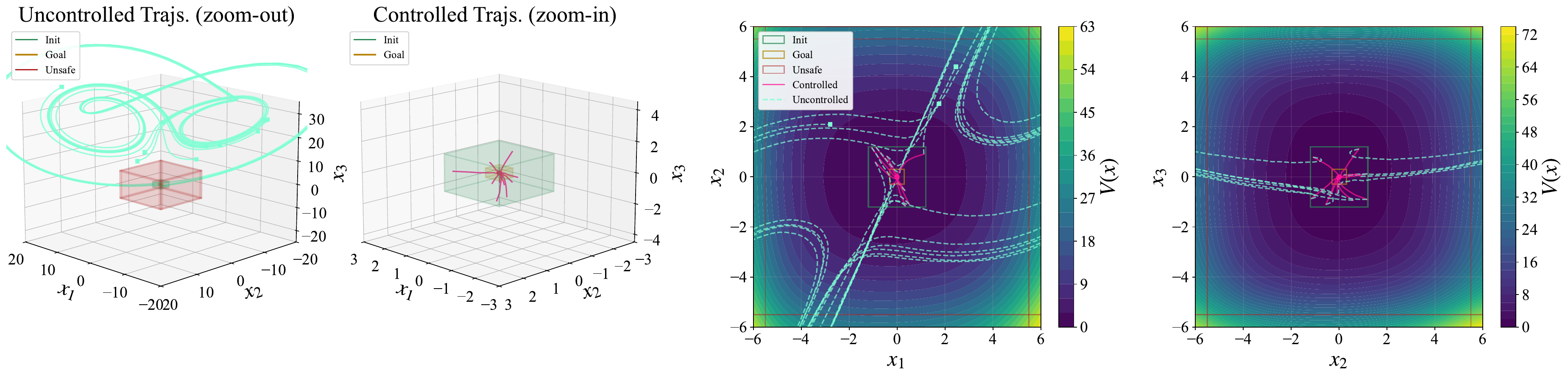}
    \caption{Lorenz synthesis: 3D trajectories (left) and 2D projections on $V_{\theta}$ contours (right), controlled vs.\ uncontrolled.}
    \label{fig:lorenz_syn_results}
\end{figure}
Figure~\ref{fig:xv15_syn_results} showcases control synthesis for the XV15 aircraft. 
We compare three strategies: constant trimmed inputs for forward flight, an initialized controller trained by sampling, and the controller synthesized via bound-training loss. 
The left panel shows the 3D paths, while the middle and right panels report the longitudinal position and control inputs. 
Under constant trim, the aircraft quickly becomes unstable (aqua dashed lines). In contrast, the synthesized controller stabilizes the trajectories, which reach the goal with smoother inputs than the sample-trained controller (pink vs.\ navy).
\begin{figure}[!h]
    \centering
    \includegraphics[width=01.0\linewidth,height=0.8\textheight,keepaspectratio]{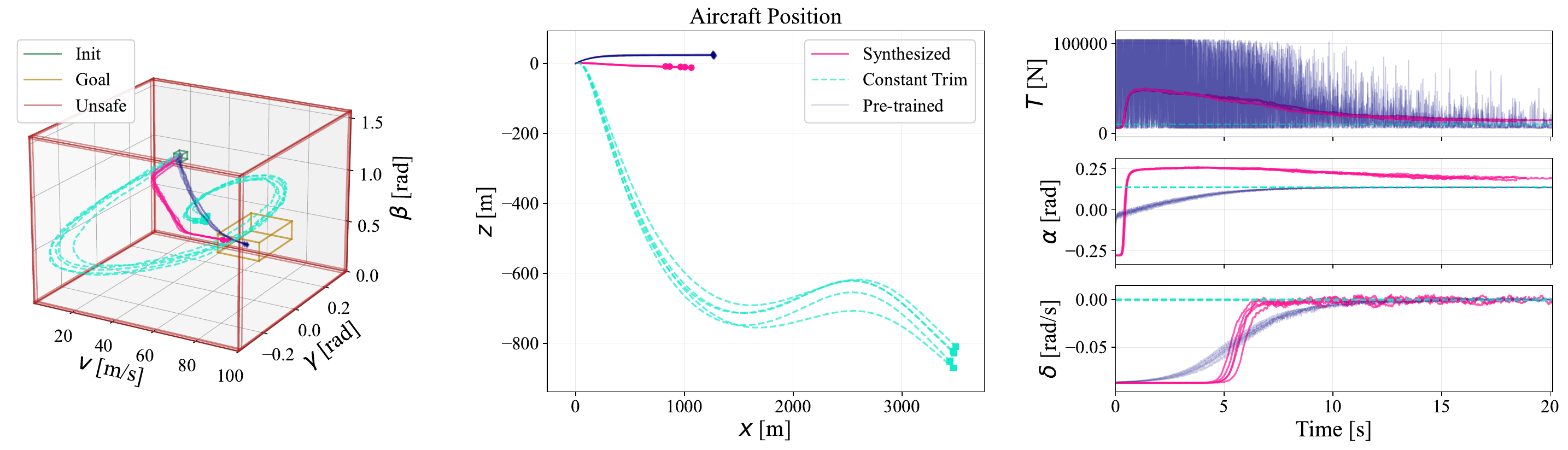}
    \caption{XV15 aircraft synthesis: phase trajectories (left), longitudinal position (middle), and control inputs (right), comparing trim, sample-trained, and synthesized controllers.}
    \label{fig:xv15_syn_results}
\end{figure}

% \begin{figure}[!h]
%   \centering
%   \begin{minipage}[t]{0.45\textwidth}
%     \centering
% \includegraphics[width=\linewidth]{figures/appendix/gbm_generator_2d.pdf}
%     \vspace{0.0em}
%     {\small (a) GBM Non-Equilibrium}
%   \end{minipage}
%   \begin{minipage}[t]{0.45\textwidth}
%     \centering
% \includegraphics[width=\linewidth]{figures/appendix/invpend_generator_2d.pdf}
%     \vspace{0.0em}
%     {\small (b) Inverted Pendulum}
%   \end{minipage}
%   \caption{Generator contour plots for 2D synthesis results.}
%   \label{fig:generator_2d}
% \end{figure}
%

\end{document}